\documentclass[onecolumn,12pt]{IEEEtran}
\usepackage{graphicx,epic,eepic,epsfig,amsmath,amsfonts,latexsym,amssymb,color,enumerate,bbm}
\allowdisplaybreaks[3]
\usepackage[colorlinks,
            linkcolor=blue,
            anchorcolor=blue,
            citecolor=red,
            ]{hyperref}

\newtheorem{theorem}{Theorem}

\newtheorem{corollary}[theorem]{Corollary}
\newtheorem{definition}[theorem]{Definition}
\newtheorem{lemma}[theorem]{Lemma}
\newtheorem{proposition}[theorem]{Proposition}
\newtheorem{remark}[theorem]{Remark}

\def\QED{\mbox{\rule[0pt]{1.5ex}{1.5ex}}}

\newcommand{\qed}{\hfill \QED}
\newcommand{\tr}{\operatorname{Tr}}
\newcommand{\ket}[1]{|#1\rangle}
\newcommand{\bra}[1]{\langle#1|}

\newcommand{\beq}{\begin{equation}}
\newcommand{\eeq}{\end{equation}}
\newcommand{\rar}{\rightarrow}

\newcommand{\1}{\mathbbm{1}}
\newcommand{\supp}{\operatorname{supp}}
\newcommand{\ox}{\otimes}
\newcommand{\nb}{\nonumber}
\newcommand{\mc}{\mathcal}
\newcommand{\mb}{\mathbb}
\newenvironment{proofof}[1]{\vspace*{5mm} \par \noindent
{\it Proof of #1:\hspace{2mm}}}{\qed
}

\begin{document}
\title{Strong Converse Exponent for Entanglement-Assisted Communication}

\author{Ke~Li, Yongsheng~Yao

\thanks{The work of Ke Li was supported by the National Natural Science Foundation
of China (No. 61871156, No. 12031004). The work of Yongsheng Yao was supported by the
National Natural Science Foundation of China (No. 61871156, No. 12071099).

Ke Li is with the Institute for Advanced Study in Mathematics,
Harbin Institute of Technology, Nangang District, Harbin 150001,
China (e-mail:carl.ke.lee@gmail.com, keli@hit.edu.cn).

Yongsheng Yao is with the Institute for Advanced Study in Mathematics and
School of Mathematics, Harbin Institute of Technology, Nangang District,
Harbin 150001, China (e-mail:yongsh.yao@gmail.com).}}
\date{}

\maketitle
\begin{abstract}
We determine the exact strong converse exponent for entanglement-assisted
classical communication of a quantum channel. Our main contribution is the
derivation of an upper bound for the strong converse exponent which is
characterized by the sandwiched R{\'e}nyi divergence. It turns out that this
upper bound coincides with the lower bound of Gupta and Wilde (Commun.\ Math.\
Phys.\ 334:867--887, 2015). Thus, the strong converse exponent follows from the
combination of these two bounds. Our result has two implications. Firstly,
it implies that the exponential bound for the strong converse property of
quantum-feedback-assisted classical communication, derived by
Cooney, Mosonyi and Wilde (Commun.\ Math.\ Phys.\ 344:797--829, 2016), is optimal.
This answers their open question in the affirmative. Hence, we have determined
the exact strong converse exponent for this problem as well. Secondly, due to
an observation of Leung and Matthews, it can be easily extended to deal with the
transmission of quantum information under the assistance of entanglement or
quantum feedback, yielding similar results. The above findings provide, for
the first time, a complete operational interpretation to the channel's
sandwiched R{\'e}nyi information of order $\alpha > 1$.
\end{abstract}


\section{Introduction}
\label{sec:introduction}
Understanding the ultimate limit that the laws of quantum mechanics impose on our
capability to communicate information is a topic of great significance.
Started in the 1970s, the investigation of the quantum Shannon theory has made huge
achievements, yet lots of important problems remain unsolved~\cite{Wilde2013quantum}.
The setting of entanglement-assisted classical communication of quantum channels
turns out to be of particular interest. It is regarded as the most natural quantum
generalization of the classical channel coding problem considered by Shannon~\cite{Shannon1948mathematical}.

The entanglement-assisted classical capacity $C_E(\mc{N})$ of a quantum channel
$\mc{N}$ quantifies the maximal rate of classical information the channel
$\mc{N}$ can reliably transmit, when encoding among multiple uses of $\mc{N}$ is
permitted, and unlimited entanglement shared by the sender and the
receiver are available~\cite{BSST1999entanglement}. The formula of $C_E(\mc{N})$
was derived by Bennett, Shor, Smolin and Thapliyal more than two decades
ago~\cite{BSST2002entanglement}, and it is given by the maximal quantum mutual
information that the channel $\mc{N}$ can generate, in formal analogy to Shannon's
formula for the capacity of classical channels. Later on, Bowen showed that feedback
provides no increase to the entanglement-assisted capacity~\cite{Bowen2004quantum},
paralleling the fact that classical feedback does not change the
capacity of classical channels~\cite{CoverThomas1991elements}. The quantum reverse
Shannon theorem established by Bennett, Devetak, Harrow, Shor and Winter
in~\cite{BDHSW2014quantum} further gives that the quantity $C_E(\mc{N})$ is the
minimal rate of classical communication needed to simulate the channel $\mc{N}$
with the assistance of unlimited entanglement (see also~\cite{BCR2011the} for an
alternative proof).

As a result of the quantum reverse Shannon theorem, we know that the
entanglement-assisted classical capacity of a quantum channel satisfies the strong
converse property~\cite{BDHSW2014quantum,BCR2011the}. That is, if one transmits
classical information at a rate larger than the capacity $C_E(\mc{N})$, then the
probability of successfully decoding the messages must decay to $0$ exponentially
fast, as the number of channel uses increases. This fact implies that $C_E(\mc{N})$
is a critical changing point for entanglement-assisted communication. The exact rate
of this exponential decay that can be achieved by the best strategy is called the
\emph{strong converse exponent}.

The formula of the strong converse exponent of quantum channels for entanglement-assisted
classical communication was unknown. This is in contrast to the fact that the strong
converse exponents of classical channels~\cite{Arimoto1973converse, DueckKorner1979reliability,
CsiszarKorner1982feedback} and classical-quantum channels~\cite{MosonyiOgawa2017strong},
have been well understood. In~\cite{GuptaWilde2015multiplicativity},
Gupta and Wilde have derived a lower bound in terms of the sandwiched R\'enyi
divergence~\cite{MDSFT2013on, WWY2014strong}, exploiting the multiplicativity of the
completely bounded $p$-norms~\cite{DJKR2006multiplicativity} as well as the
``entanglement-assisted meta-converse'' property~\cite{MatthewsWehner2014finite}.
Subsequently, Cooney, Mosonyi and Wilde~\cite{CMW2016strong} have strengthened this
result, showing that the lower bound obtained in~\cite{GuptaWilde2015multiplicativity}
still holds even if additional classical or quantum feedback from the receiver to the
sender is allowed. Note that a lower bound for the strong converse exponent translates
to an upper bound for the success probability.

In this paper, we derive an upper bound for the strong converse exponent of
entanglement-assisted classical communication. Our upper bound coincides with the lower
bound of Gupta and Wilde~\cite{GuptaWilde2015multiplicativity}. Thus, the combination
of these two bounds lets us completely determine the exact strong converse exponent of  entanglement-assisted classical communication. We construct a two-step proof, following
the work of Mosonyi and Ogawa~\cite{MosonyiOgawa2017strong}. In the
first step, we prove a Dueck-K\"orner-type~\cite{DueckKorner1979reliability} upper
bound which is expressed using the quantum relative entropy, and we show that it is
equivalent to an Arimoto-type~\cite{Arimoto1973converse} bound as a transform of the
log-Euclidean R{\'e}nyi information. In the second step, we employ multiple techniques
to turn the suboptimal bound derived in the first step into the final one, which is
of the similar form, but now the log-Euclidean R{\'e}nyi information is replaced by
the sandwiched R\'enyi information. Compared to the work~\cite{MosonyiOgawa2017strong}
which deals with the classical-quantum channels, we need to develop new methods to
overcome the technical difficulties in coping with entanglement-assisted communication
over general quantum channels. On the one hand, we introduce a type of the log-Euclidean
R{\'e}nyi information of a quantum channel, which is defined with respect to an
ensemble of input states, and this quantity will play an important role. One the
other hand, we treat the eigenvalues and the eigenvectors of the input density
matrices separately and we make several subtle uses of the minimax theorem.

As a corollary, we conclude that the lower bound of Cooney, Mosonyi and
Wilde~\cite{CMW2016strong} for the strong converse exponent of quantum-feedback-assisted
classical communication is optimal. So, we have determined the exact strong converse
exponent for this problem as well. This shows that additional feedback does not affect
the strong converse exponent of entanglement-assisted classical communication.

While we have focused on the transmission of classical information, our results can
be easily extended to deal with the transmission of quantum information. When free
entanglement is available, Leung and Matthews~\cite{LeungMatthews2015on} have shown
a quantitative equivalence between noisy classical communication and noisy quantum
communication, generalizing the interchange between perfect classical communication
and perfect quantum communication via teleportation~\cite{BBCJPW1993teleporting}
and dense coding~\cite{BennettWiesner1992communication}. This lets us obtain the
strong converse exponent for quantum communication under the assistance of
entanglement or quantum feedback.

We point out that entanglement-assisted communication has drawn much attention in the
community of quantum information. In addition to the results mentioned above, we review
some recent developments. One-shot characterizations of the entanglement-assisted
capacity of a quantum channel were given, e.g., in~\cite{DattaHsieh2012one,
MatthewsWehner2014finite, AJW2018building}. The second-order asymptotics was addressed
in~\cite{DTW2016on} , where an achievability bound was obtained. Progress on the
achievability part of the direct error exponent can be seen
in~\cite{QWW2018applications, Cheng2022a}. The moderate deviation expansion in the
high-error regime has been derived in~\cite{RNTB2023moderate}. The full resolutions of
these asymptotic characterizations are interesting open problems.

The remainder of this paper is organized as follows. In Section~\ref{sec:preliminaries}
we introduce the necessary notation, definitions and some basic properties. In
Section~\ref{sec:problem-results} we present the main problems and results. In Section~\ref{sec:proof-interm} we prove an intermediate upper bound for the strong
converse exponent of entanglement-assisted classical communication. Then in
Section~\ref{sec:proof-final} we improve it to obtain the final result. At last,
in Section~\ref{sec:discussion} we conclude the paper with some discussion.

\section{Preliminaries}
\label{sec:preliminaries}
\subsection{Notation and basic properties}
For a Hilbert space $\mc{H}$, we denote by $\mc{L}(\mc{H})$ the set of linear operators
on it, and we use $\mc{L}(\mc{H})_+$ for the set of positive semidefinite operators.
Density operators, or quantum states, are positive semidefinite operators with trace $1$.
Let $\mc{H}_A$ be the Hilbert space associated with a quantum system $A$. The
set of quantum states on $\mc{H}_A$ is denoted by $\mc{S}(\mc{H}_A)$, or $\mc{S}(A)$ for
short. Pure states are denoted by $\mc{S}_1(\mc{H}_A)$ or $\mc{S}_1(A)$.
$\1_A$ and $\pi_A$ are the identity operator and the maximally mixed state on $\mc{H}_A$,
respectively. The support of a positive semidefinite operator $X$ is denoted
by $\supp(X)$. For $\rho\in\mc{S}(A)$, $\mc{S}_\rho(A)$ represents the set of quantum
states whose support are included in $\supp(\rho)$. Throughout this paper, we are restricted
to quantum systems of finite dimension. We denote as $|A|$ the dimension of $\mc{H}_A$.

Let $\mc{X}$ denote a finite alphabet set. For the sequence $x^n\in\mc{X}^{\times n}$,
the type~\cite{Csiszar1998method} $P_{x^n}$ is defined as the empirical distribution of
$x^n$, i.e.,
\begin{equation}
\label{type:1}
P_{x^n}(a)=\frac{1}{n} \sum_{i=1}^{n} \delta_{x_i,a}, \quad \forall a \in \mc{X}.
\end{equation}
The notation $\mc{T}_n$ is used for the set of all types. We can bound the number of
types as
\begin{equation}
\label{type:2}
|\mc{T}_n| \leq (n+1)^{|\mc{X}|}.
\end{equation}
For $t \in \mc{T}_n$, we denote by $T_n^t$  the set of all sequences of type $t$, i.e.,
\begin{equation}
\label{type:3}
T_n^t:=\{x^n ~|~ P_{x^n}=t \}.
\end{equation}
For a pure state $\ket{\psi}_{AA'}$ with Schmidt decomposition $\ket{\psi}_{AA'}=\sum_{x=1}^{|A|} \sqrt{p(x)}\ket{a_x}_A \ox \ket{a_x}_{A'}$, we can write $\ket{\psi}_{AA'}^{\ox n}$ as
\begin{equation}
\label{equ:typedecompose}
\ket{\psi}_{AA'}^{\ox n}=\sum_{t \in \mc{T}_n} \sqrt{p^n(t)} \ket{\Psi^t}_{A^nA'^n},
\end{equation}
where $\mc{T}_n$ is the set of types with respect to the alphabet set $\mc{X}=\{1,2,\ldots,|A|\}$,
\begin{align}
&p^n(t)=\sum_{x^n\in T_n^t}p(x_1)p(x_2)\cdots p(x_n), \quad\text{and}  \\
&\ket{\Psi^t}_{A^nA'^n}=\frac{1}{\sqrt{|T_n^t|}}
\sum_{x^n\in T_n^t}\ket{a^n_{x^n}}_{A^n}\ox \ket{a^n_{x^n}}_{A'^n}
\end{align}
with $\ket{a^n_{x^n}}=\ket{a_{x_1}}\ox\ket{a_{x_2}}\ox\cdots\ox\ket{a_{x_n}}$.

A quantum channel (or quantum operation) $\mc{N}$ is a linear, completely positive and trace-preserving~(CPTP) map. We denote by $\mc{N}_{A \rar B}$ a quantum channel
$\mc{N}: \mc{L}(\mc{H}_A) \rar \mc{L}(\mc{H}_B)$. A quantum measurement is represented by
a set of positive semidefinite operators $\{M_x\}_x$ such that $\sum_xM_x=\1$.
When making this measurement on a system in the state $\rho$, we get outcome $x$
with probability $\tr\rho M_x$.

Let $\sigma \in \mc{L}(\mc{H})$ be a self-adjoint operator. We denote the number of
distinct eigenvalues of $\sigma$ as $v(\sigma)$. Let the spectral projections of $\sigma$
be $P_1, \ldots, P_{v(\sigma)}$. The pinching map $\mc{P}_\sigma$ associated with
$\sigma$ is given by
\beq
\mc{P}_\sigma: X \mapsto \sum^{v(\sigma)}_{i=1} P_i X P_i.
\eeq
The pinching inequality~\cite{Hayashi2002optimal} says that for any positive semidefinite
operator $X$, we have
\begin{equation}
\label{eq:pinchingine}
X \leq v(\sigma) \mc{P}_\sigma(X).
\end{equation}

Let $S_n$ be the group of permutations over a set of $n$ elements. The natural representation
of $S_n$ on $\mc{H}_B^{\ox n}$ is given by the unitary transformations
\beq\label{eq:permutation}
W^{\iota}_{B^n} : \ket{\psi_1} \ox \ldots \ox \ket{\psi_n}
\mapsto  \ket{\psi_{\iota^{-1}(1)}} \ox \ldots \ox \ket{\psi_{\iota^{-1}(n)}},
\quad \ket{\psi_i} \in \mc{H}_B, \iota \in S_n.
\eeq
We denote by $\mc{S}_{\text{sym}}(B^n)$ the set of symmetric states on $\mc{H}_B^{\ox n}$.
That is,
\beq
\mc{S}_{\text{sym}}(B^n):=\{\sigma_{B^n} \in \mc{S}(B^n)~|~
W^{\iota}_{B^n} \sigma_{B^n}W^{\iota \dagger}_{B^n}=\sigma_{B^n}, \ \forall\ \iota \in S_n\},
\eeq
where $W^{\iota \dagger}_{B^n}$ is the adjoint of $W^{\iota}_{B^n}$. There exists a
symmetric state that dominates all the others, as stated in Lemma~\ref{lem:u-sym} below.
Two different constructions are given by~\cite{CKR2009postselection}
and~\cite{Hayashi2009universal}, respectively.
Detailed arguments can be found in~\cite[Lemma 1]{HayashiTomamichel2016correlation}
and~\cite[Appendix A]{MosonyiOgawa2017strong}.
\begin{lemma}
\label{lem:u-sym}
For every finite-dimensional system $B$ and every $n \in \mb{N}$, there exists a symmetric state
$\sigma_{B^n}^u \in \mc{S}_{\rm{sym}}(B^n)$ such that every symmetric state $\sigma_{B^n} \in \mc{S}_{\rm{sym}}(B^n)$ is dominated by $\sigma_{B^n}^u$ as
\begin{equation}
\sigma_{B^n} \leq v_{n,|B|}\sigma_{B^n}^u,
\end{equation}
where $v_{n,|B|} \leq (n+1)^{\frac{(|B|+2)(|B|-1)}{2}}$ is a polynomial of $n$. The number of
different eigenvalues of $\sigma_{B^n}^u$ is also upper bounded by $v_{n,|B|}$.
\end{lemma}

\subsection{Quantum R{\'e}nyi divergences}
The classical R{\'e}nyi divergence~\cite{Renyi1961measures} is a crucial information
quantity. In the quantum domain, there can be many inequivalent generalizations of the
R{\'e}nyi divergence. In this paper, we are concerned with the \emph{sandwiched R{\'e}nyi
divergence} and the \emph{log-Euclidean R{\'e}nyi divergence}. The sandwiched R\'enyi
divergence was introduced in~\cite{MDSFT2013on} and~\cite{WWY2014strong}, and it is one
of the only two quantum R{\'e}nyi divergences that have found operational meanings (the
other one is Petz's R{\'e}nyi divergence~\cite{Petz1986quasi}).

\begin{definition}
\label{def:sandwiched-RD}
Let $\alpha \in (0,1)\cup(1,+\infty)$, and let $\rho\in\mc{S}(\mc{H})$ be a quantum
state and $\sigma\in\mc{L}(\mc{H})_+$ be positive semidefinite. When $\alpha > 1$ and $\supp(\rho)\subseteq\supp(\sigma)$ or $\alpha \in (0,1)$ and $\supp(\rho)\not\perp\supp(\sigma)$,
the sandwiched R{\'e}nyi divergence of order $\alpha $ is defined as
\beq
D_{\alpha}^*(\rho \| \sigma)
:=\frac{1}{\alpha-1} \log \tr {({\sigma}^{\frac{1-\alpha}{2\alpha}} \rho {\sigma}^{\frac{1-\alpha}{2\alpha}})}^\alpha;
\eeq
otherwise, we set $D_{\alpha}^*(\rho \| \sigma)=+\infty$.
\end{definition}

The expression of the log-Euclidean R{\'e}nyi divergence appeared already
in~\cite{HiaiPetz1993the} and~\cite{OgawaNagaoka2000strong}, for certain range of
the R{\'e}nyi parameter. However, its significance has been recognized only quite
recently in~\cite{MosonyiOgawa2017strong}, where it has been thoroughly studied and was
for the first time employed as an intermediate quantity in the derivation of the strong
converse exponent for classical-quantum channels. This name was first called
in~\cite{CHT2019quantum}.
\begin{definition}
\label{def:LogE-RD}
Let $\alpha \in (0,1)\cup(1,+\infty)$, and let $\rho\in\mc{S}(\mc{H})$ be a quantum
state and $\sigma\in\mc{L}(\mc{H})_+$ be positive semidefinite. The log-Euclidean
R{\'e}nyi divergence of order $\alpha$ is defined as
\beq\label{eq:def-logE-RD-1}
D_{\alpha}^{\flat}(\rho \| \sigma)
:=\frac{1}{\alpha-1} \log \tr 2^{\alpha\log\rho+(1-\alpha)\log\sigma}
\eeq
if $\rho$ and $\sigma$ are of full rank; otherwise, we replace the expression in
the logarithm by
\beq\label{eq:def-logE-RD-2}
\lim_{\epsilon\searrow 0}\tr 2^{\alpha\log(\rho+\epsilon\1)
+(1-\alpha)\log(\sigma+\epsilon\1)}.
\eeq
\end{definition}

\begin{remark}
The existence of the limit in Eq.~\eqref{eq:def-logE-RD-2} was shown
in~\cite[Lemma 3.1]{MosonyiOgawa2017strong}, where an alternative expression was
given. The condition for $D_{\alpha}^{\flat}(\rho \| \sigma)=+\infty$ is
$\alpha > 1$ and $\supp(\rho)\nsubseteq\supp(\sigma)$ or $\alpha \in (0,1)$ and
$\supp(\rho)\cap\supp(\sigma)=\{0\}$; see~\cite[Remark 3.4]{MosonyiOgawa2017strong}.
\end{remark}

For a bipartite state $\rho_{AB} \in \mc{S}(AB)$ and $\alpha \in (0,1)\cup(1,+\infty)$, the
sandwiched R{\'e}nyi mutual information is defined as~\cite{WWY2014strong, Beigi2013sandwiched}
\begin{equation}
\label{equ:mutualinf}
I_{\alpha}^*(A:B)_\rho=\min_{\sigma_B} D_{\alpha}^*(\rho_{AB} \| \rho_A \ox \sigma_B).
\end{equation}

When $\alpha$ goes to $1$, both $D_{\alpha}^*$ and $D_{\alpha}^{\flat}$ converge to the quantum
relative entropy~\cite{Umegaki1954conditional}
\beq
  D(\rho\|\sigma):= \begin{cases}
                      \tr(\rho(\log\rho-\log\sigma)) & \text{ if }\supp(\rho)\subseteq\supp(\sigma), \\
                      +\infty                        & \text{ otherwise.}
                    \end{cases}
\eeq
Similarly, the quantum mutual information for $\rho_{AB} \in \mc{S}({AB})$, defined as
\beq
I(A:B)_\rho :=D(\rho_{AB} \| \rho_A \ox \rho_B),
\eeq
is the limiting case of the sandwiched R{\'e}nyi mutual information. So, we extend the
definition of the R{\'e}nyi information quantities to include the case $\alpha=1$ by
taking the limits. We refer the readers to~\cite{MDSFT2013on,WWY2014strong,MosonyiOgawa2017strong}
for the computation of these limits.

An ensemble of quantum states is a set of pairs $\{p_x,\rho_x\}_{x\in\mc{X}}$, where
$\{p_x\}_{x\in\mc{X}}$ is a probability distribution and $\{\rho_x\}_x$ are quantum
states. Its Holevo information is defined as
\beq
\chi\big(\{p_x,\rho_x\}_x\big):=\sum_xp_xD\big(\rho_x\|\sum_xp_x\rho_x\big).
\eeq

In the following proposition, we collect some properties of the R{\'e}nyi information
quantities.

\begin{proposition}
\label{prop:renyid}
Let $\rho \in \mc{S}(\mc{H})$ be a quantum state and $\sigma \in \mc{L}(\mc{H})_+$ be
positive semidefinite. The sandwiched R{\'e}nyi divergence and the log-Euclidean
R{\'e}nyi divergence satisfy the following properties.
\begin{enumerate}[(i)]
  \item Monotonicity in R{\'e}nyi
        parameter~\cite{Beigi2013sandwiched,MDSFT2013on,MosonyiOgawa2017strong}:
        if $0\leq \alpha \leq \beta$, then $D_{\alpha}^{(t)}(\rho \| \sigma) \leq
        D^{(t)}_{\beta}(\rho \| \sigma)$, for $(t)=*$ and $(t)=\flat$;
  \item Monotonicity in $\sigma$~\cite{MDSFT2013on,MosonyiOgawa2017strong}:
        if $\sigma' \geq \sigma$, then $D_{\alpha}^{(t)}(\rho \| \sigma')
        \leq D_{\alpha}^{(t)}(\rho \| \sigma)$, for $(t)=*$, $\alpha \in [\frac{1}{2},+\infty)$
        and for $(t)=\flat$, $\alpha \in [0,+\infty)$;
  \item Variational representation~\cite{MosonyiOgawa2017strong}:
        the log-Euclidean R{\'e}nyi divergence has the following variational representation
        \begin{equation}
          D_{\alpha}^{\flat}(\rho \| \sigma)=s(\alpha) \max_{\tau \in \mc{S}_\rho(\mc{H})} s(\alpha)
          \{D(\tau \| \sigma)-\frac{\alpha}{\alpha-1}D(\tau \| \rho)\},
        \end{equation}
        where $s(\alpha)=1$ for $\alpha \in (1,+\infty)$ and $s(\alpha)=-1$ for $\alpha \in (0,1)$, and $\mc{S}_\rho(\mc{H})$ denotes the set of states whose supports are contained in the support of $\rho$;
  \item Additivity of sandwiched R{\'e}nyi mutual information~\cite{HayashiTomamichel2016correlation}:
        for two states $\rho_{AB}\in\mc{S}(AB)$ and $\sigma_{A'B'} \in \mc{S}(A'B')$, and for any
        $\alpha \in [\frac{1}{2}, +\infty)$, we have
        \beq
        I_{\alpha}^*(AA':BB')_{\rho \ox \sigma}=I_{\alpha}^*(A:B)_\rho+I_{\alpha}^*(A':B')_\sigma ;
        \eeq
  \item Convexity in $\sigma$~\cite{MosonyiOgawa2017strong}:
        the function $\sigma\mapsto D_{\alpha}^{(t)}(\rho \| \sigma)$ is convex
        for $(t)=*$, $\alpha \in [\frac{1}{2},+\infty)$ and for $(t)=\flat$, $\alpha \in [0,+\infty)$;
  \item Approximation by pinching~\cite{MosonyiOgawa2015quantum,HayashiTomamichel2016correlation}:
        for $\alpha\geq 0$, we have
        \beq
        D^*_\alpha(\mc{P}_\sigma(\rho)\|\sigma)
        \leq D^*_\alpha(\rho\|\sigma)
        \leq D^*_\alpha(\mc{P}_\sigma(\rho)\|\sigma)+2\log v(\sigma).
        \eeq
\end{enumerate}
\end{proposition}

\section{Problems and Results}
\label{sec:problem-results}
\subsection{Entanglement-assisted classical communication}
Suppose that the sender Alice and the receiver Bob are connected by a quantum channel
$\mc{N}_{A \rar B}$, and they share arbitrary entangled quantum states. Alice wants to
send classical messages to Bob. We start with a description of a code for a single use
of the channel $\mc{N}_{A \rar B}$. Let $\mc{M}=\{1,\ldots , M\}$ be the set of messages
to be transmitted. Let $\rho_{\tilde{A}\tilde{B}}$ be the entangled state shared by
Alice ($\tilde{A}$) and Bob ($\tilde{B}$). To send the message $m\in\mc{M}$, Alice
applies a CPTP map $\mc{E}^m_{\tilde{A} \rar A}$ to her half of the state
$\rho_{\tilde{A}\tilde{B}}$. Then she inputs the $A$ system to the channel
$\mc{N}_{A \rar B}$. After receiving the channel output $B$, Bob performs a decoding
measurement $\{\Lambda^m_{B\tilde{B}}\}_{m \in \mc{M}}$ on the system $B\tilde{B}$ to
recover the classical message. The collection
\beq
\mc{C}\equiv\big(\{\mc{E}_{\tilde{A}\rar A}^m\}_{m \in \mc{M}},
\{\Lambda^m_{B\tilde{B}}\}_{m \in \mc{M}}, \rho_{\tilde{A}\tilde{B}}\big)
\eeq
is called a code for $\mc{N}_{A \rar B}$. The size of the code is $|\mc{C}|:=M$, and
the success probability of $\mc{C}$ is
\begin{equation}\label{eq:Ps}
P_s(\mc{N}_{A \rar B},\mc{C}):=\frac{1}{M} \sum_{m=1}^{M} \tr \mc{N}_{A \rar B}
\circ\mc{E}^m_{\tilde{A}\rar A}(\rho_{\tilde{A}\tilde{B}})\Lambda^m_{B\tilde{B}}.
\end{equation}

The art in information theory is to make codes for multiple uses of the channel. We denote
a code for $\mc{N}_{A \rar B}^{\ox n}$ by $\mc{C}_n$. The channel capacity is the maximal
rate of information transmission that can be achieved by a sequence of codes $\{\mc{C}_n\}_
{n \in \mb{N}}$, under the condition that the success probability goes to $1$ asymptotically.
Formally, we define the entanglement-assisted classical capacity as
\beq
C_E(\mc{N}_{A \rar B})
:=\sup\left\{\liminf_{n\rar\infty} \frac{1}{n} \log |\mc{C}_n| ~\Big |~
\lim_{n \rar \infty} P_s(\mc{N}_{A \rar B}^{\ox n},\mc{C}_n)=1 \right\}.
\eeq
It has been established in~\cite{BSST2002entanglement} that
\beq\label{eq:entascap}
C_E(\mc{N}_{A \rar B})
=\max_{\psi_{A'A} \in \mc{S}_1(A'A)} I(A':B)_{\mc{N}_{A \rar B}(\psi_{A'A})}.
\eeq
Note that $\mc{S}_1(A'A)$ is the set of pure states, and without loss of generality, we
can assume that $|A'|=|A|$ in Eq.~\eqref{eq:entascap}.

\subsection{Main result}
It is known that when the transmission rate is larger than the capacity $C_E(\mc{N}_{A \rar B})$,
the success probability $P_s(\mc{N}_{A \rar B}^{\ox n}, \mc{C}_n)$ inevitably converges to $0$ as $n\rar\infty$~\cite{BDHSW2014quantum}. This is called the strong converse property.
The strong converse exponent characterizes the speed of this convergence. It is defined
as the best rate of exponential decay of the success probability, for a fixed transmission
rate $R$:
\begin{equation}
  \label{eq:def-scee}
sc(\mc{N}_{A \rar B},R)
:=\inf\left\{-\liminf_{n\rar\infty}\frac{1}{n}\log P_s(\mc{N}_{A \rar B}^{\ox n},\mc{C}_n)
~\Big|~ \liminf_{n\rar\infty} \frac{1}{n} \log |\mc{C}_n| \geq R \right\}.
\end{equation}

Gupta and Wilde~\cite{GuptaWilde2015multiplicativity} have proven that
\begin{equation}
\label{eq:gupta-wilde}
sc(\mc{N}_{A \rar B},R)
\geq \sup_{\alpha>1}\frac{\alpha-1}{\alpha}\big(R-I_{\alpha}^*(\mc{N}_{A \rar B})\big),
\end{equation}
where $I_{\alpha}^*(\mc{N}_{A \rar B})$ is the sandwiched R{\'e}nyi information of the
channel $\mc{N}_{A \rar B}$, defined as
\begin{equation}
\label{equ:channelmut}
I_{\alpha}^*(\mc{N}_{A \rar B})
:=\max_{\psi_{A'A}\in\mc{S}_1(A'A)} I_{\alpha}^*(A':B)_{\mc{N}_{A \rar B}(\psi_{A'A})}.
\end{equation}
Without loss of generality, we assume that $|A'|=|A|$ in Eq.~\eqref{equ:channelmut}.

The main contribution of the present paper is to prove the other direction. As a result,
we have established the following equality.

\begin{theorem}
\label{thm:main}
Let $\mc{N}_{A \rar B}$ be a quantum channel. For any transmission rate $R>0$, the
strong converse exponent for entanglement-assisted classical communication is
\begin{equation}
\label{eq:scexponent}
sc(\mc{N}_{A \rar B},R)
= \sup_{\alpha>1} \frac{\alpha-1}{\alpha}\big(R-I_{\alpha}^*(\mc{N}_{A \rar B})\big).
\end{equation}
\end{theorem}

Theorem~\ref{thm:main} has two implications as illustrated in the next two subsections.
This enables us to also determine the strong converse exponents for quantum-feedback-assisted
communication and for transmitting quantum information in the same settings. The
proof of Theorem~\ref{thm:main} is given in Section~\ref{sec:proof-interm} and Section~\ref{sec:proof-final}.

\subsection{Implication to quantum-feedback-assisted communication}
For a quantum channel $\mc{N}_{A \rar B}$, quantum-feedback-assisted classical communication
is the transmission of classical information with the assistance of noiseless quantum
feedback from the receiver to the sender~\cite{Bowen2004quantum,CMW2016strong}. Since
quantum feedback can generate arbitrary bipartite entangled state, it is a stronger
resource than shared entanglement. Indeed, using quantum teleportation~\cite{BBCJPW1993teleporting},
we can easily see that quantum feedback is equivalent to shared entanglement plus
classical feedback.

The quantum-feedback-assisted classical capacity of a quantum channel
$\mc{N}_{A \rar B}$ is still $C_E(\mc{N}_{A \rar B})$~\cite{Bowen2004quantum}. When the
transmission rate is larger than $C_E(\mc{N}_{A \rar B})$, the strong converse property
was proved in~\cite{BDHSW2014quantum} via the technique of channel simulation. The exact
strong converse exponent was not known, however. We point out that, a general
quantum-feedback-assisted code can be very complicated, where the encoding can adaptively
depend on previous feedbacks. We refer to~\cite[Section 5]{CMW2016strong} for an explicit
description of the quantum-feedback-assisted codes. The strong converse exponent in
this setting is defined similar to Eq.~\eqref{eq:def-scee}, and we denote it by
$sc^{\leftarrow}(\mc{N}_{A \rar B},R)$.

Cooney, Mosonyi and Wilde have proved in~\cite{CMW2016strong} that the right hand side of
Eq.~\eqref{eq:scexponent} is a lower bound for $sc^{\leftarrow}(\mc{N}_{A \rar B},R)$ as well.
As an open question, they ask whether it is optimal. Having shown that it is already
achievable under the assistance of shared entanglement, we answer their question in the
affirmative.

\begin{corollary}
\label{cor:feedback}
Let $\mc{N}_{A \rar B}$ be a quantum channel. For any transmission rate $R>0$, the strong
converse exponent for quantum-feedback-assisted classical communication is
\begin{equation}
\label{eq:scexponentf}
sc^{\leftarrow}(\mc{N}_{A \rar B},R)
= \sup_{\alpha>1}\frac{\alpha-1}{\alpha}\big(R-I_{\alpha}^*(\mc{N}_{A \rar B})\big).
\end{equation}
\end{corollary}

Corollary~\ref{cor:feedback} is an immediate consequence of~\cite[Theorem 4]{CMW2016strong}
and our Theorem~\ref{thm:main}, noting that by definition
$sc^{\leftarrow}(\mc{N}_{A \rar B},R)\leq sc(\mc{N}_{A \rar B},R)$ is obvious.

\subsection{Implication to quantum information transmission}
The above results can be extended to the case of quantum information transmission.
Let us recall an entanglement-assisted code for a single use of the channel
$\mc{N}_{A \rar B}$, for transmitting quantum information that is stored in a
quantum system $M$. It consists of using a shared entangled state
$\rho_{\tilde{A}\tilde{B}}$ between the sender Alice and the receiver Bob,
applying local operation $\mc{E}_{M\tilde{A} \rightarrow A}$ at Alice's side,
feeding the system $A$ into $\mc{N}_{A \rightarrow B}$, and at last applying
local operation $\mc{D}_{B\tilde{B} \rightarrow M}$ at Bob's side. We denote by
\beq
\overline{\mc{C}}\equiv\big(\mc{E}_{M\tilde{A} \rightarrow A},
\mc{D}_{B\tilde{B} \rightarrow M}, \rho_{\tilde{A}\tilde{B}}\big)
\eeq
the code for $\mc{N}_{A \rar B}$. Its size is $|\overline{\mc{C}}|:=|M|$,
and its performance is quantified by the entanglement fidelity
\begin{equation}
P_f(\mc{N}_{A \rar B},\overline{\mc{C}})
:=F\big(\mc{D}_{B\tilde{B} \rar M}\circ \mc{N}_{A \rar B}\circ \mc{E}_{M\tilde{A}
\rightarrow A}(\Psi_{M'M} \ox \rho_{\tilde{A}\tilde{B}}), \Psi_{M'M}\big),
\end{equation}
where $\Psi_{M'M}$ is a maximally entangled pure state and $F(\rho,\sigma)=
\|\sqrt{\rho}\sqrt{\sigma}\|_1$ is the fidelity. We remark that the entanglement
fidelity $P_f$ can be regarded as an analogue of the success probability of
Eq.~\eqref{eq:Ps}. Denote a code for $\mc{N}_{A \rar B}^{\ox n}$ by
$\overline{\mc{C}}_n$. Then the strong converse exponent for entanglement-assisted
quantum communication is defined similar to Eq.~\eqref{eq:def-scee}---with $\mc{C}_n$
replaced by $\overline{\mc{C}}_n$ and $P_s$ replaced by $P_f$---and we denote it
as $sc^\text{q}(\mc{N}_{A \rightarrow B},R)$.

Let $P_f^{\star}(\mc{N}_{A \rar B}, k)$ be the optimal performance among all
entanglement-assisted codes for quantum information transmission with size $k$, and
$P_s^{\star}(\mc{N}_{A \rar B}, k)$ be the optimal success probability among all
entanglement-assisted codes for classical information transmission with size $k$.
Leung and Matthews~\cite[Appendix B]{LeungMatthews2015on} have derived the exact
relation
\begin{equation}
\label{eq:pf-ps}
(P_f^{\star}(\mc{N}_{A \rar B}, k))^2=P_s^{\star}(\mc{N}_{A \rar B}, k^2).
\end{equation}
The combination of Theorem~\ref{thm:main} and Eq.~(\ref{eq:pf-ps}) directly
gives the formula for $sc^\text{q}(\mc{N}_{A \rar B},R)$.

\begin{corollary}
\label{cor:quantum}
Let $\mc{N}_{A \rar B}$ be a quantum channel. For any transmission rate $R>0$, the
strong converse exponent for entanglement-assisted quantum communication is
\begin{equation}
\label{eq:scexponentq}
sc^\text{q}(\mc{N}_{A \rar B},R)
= \sup_{\alpha>1}\frac{\alpha-1}{\alpha}\big(R-\frac{1}{2}I_{\alpha}^*(\mc{N}_{A \rar B})\big).
\end{equation}
\end{corollary}

If the shared entanglement is replaced by quantum feedback, the strong converse
exponent for quantum information transmission is still given by the right hand side
of Eq.~\eqref{eq:scexponentq}, which can be verified by a similar argument.

\section{An Upper Bound via the Log-Euclidean R\'enyi Information}
\label{sec:proof-interm}
In this section, we derive an intermediate upper bound for the strong converse exponent
of entanglement-assisted classical communication. It employs a type of R\'enyi information
of quantum channels, in terms of the log-Euclidean R{\'e}nyi divergence. This upper
bound serves as a first step for the proof of Theorem~\ref{thm:main}. The full proof
will be accomplished in Section~\ref{sec:proof-final}.

For brevity of expression, from now on we change the notation a little to use $A'$
to label the input system of the channel $\mc{N}$. The symbol $A$ is instead used to
label a bipartite state $\psi_{AA'}$ whose $A'$ part is to be acted on by the channel
$\mc{N}_{A' \rar B}$, resulting in a bipartite state on systems $A$ and $B$.

\begin{definition}
For the quantum channel $\mc{N}_{A' \rar B}$ and an arbitrary ensemble of bipartite quantum
states $\{q(t),\psi_{AA'}^t\}_{t\in\mc{T}}$ with $q$ a probability distribution, we define
\begin{equation}
\label{eq:mic1-logE}
I_{\alpha}^{\flat}\left(\mc{N}_{A' \rar B},\{q(t),\psi_{AA'}^t\}_t\right)
:= \sum_{t \in \mc{T}} q(t)\min_{\sigma_B\in\mc{S}(B)}
D_{\alpha}^{\flat}\big(\mc{N}_{A' \rar B}(\psi_{AA'}^t) \| \psi_A^t\ox \sigma_B\big).
\end{equation}
When the quantum state ensemble consists of only a single state $\psi_{AA'}$, we also
write
\begin{equation}
\label{eq:mic1-logE-s}
I_{\alpha}^{\flat}\left(\mc{N}_{A' \rar B},\psi_{AA'}\right)
:=\min_{\sigma_B\in\mc{S}(B)}
D_{\alpha}^{\flat}\big(\mc{N}_{A' \rar B}(\psi_{AA'}) \| \psi_A\ox \sigma_B\big).
\end{equation}
\end{definition}

We always use $\Psi$ to denote the maximally entangled state on two isomorphic Hilbert
spaces $\mc{H}$ and $\mc{H'}$, whose vector form can be written as
\beq
\ket{\Psi}_{\mc{H}\mc{H}'}=\frac{1}{\sqrt{|\mc{H}|}}
\sum_{x=1}^{|\mc{H}|} \ket{a_x}_\mc{H} \ox \ket{a_x}_{\mc{H}'},
\eeq
where $\{\ket{a_x}\}_x$ is a pre-fixed orthonormal basis.

\medskip
We first prove the following upper bound. Then we will further improve it to obtain the
main result of this section (Theorem~\ref{thm:strong-euc}).

\begin{proposition}
\label{prop:weak-euc}
Let $\mc{H}_A\cong\mc{H}_{A'}$, and let $\mc{H}_A=\oplus_{t \in \mc{T}} \mc{H}_A^t$ and
$\mc{H}_{A'}=\oplus_{t \in \mc{T}} \mc{H}_{A'}^t$ be decompositions of $\mc{H}_A$ and
$\mc{H}_{A'}$ into orthogonal subspaces with $\mc{H}_A^t\cong\mc{H}_{A'}^t$. Let
$\{q(t),\Psi_{AA'}^t\}_{t\in\mc{T}}$ be an ensemble of quantum states with
$\Psi_{AA'}^t$ the maximally entangled state on $\mc{H}_A^t \ox \mc{H}_{A'}^t$.
For the channel $\mc{N}_{A' \rar B}$ and any $R>0$, the strong converse exponent
for entanglement-assisted classical communication satisfies
\beq\label{eq:weak-euc}
sc(\mc{N}_{A' \rar B},R)
\leq \sup_{\alpha>1} \frac{\alpha-1}{\alpha}
\left\{R-I_{\alpha}^{\flat}\left(\mc{N}_{A'\rar B},\{q(t),\Psi_{AA'}^t\}_t\right)\right\}.
\eeq
\end{proposition}

Before proving Proposition~\ref{prop:weak-euc}, we derive a variational expression for
the right hand side of Eq.~\eqref{eq:weak-euc}. This is given in Proposition~\ref{prop:varF}
in a more general form. For brevity, we introduce the shorthand
\beq
F\left(\mc{N}_{A' \rar B},R,\{q(t),\psi_{AA'}^t\}_t\right)
\equiv\sup_{\alpha>1} \frac{\alpha-1}{\alpha}
\left\{R-I_{\alpha}^{\flat}\left(\mc{N}_{A'\rar B},\{q(t),\psi_{AA'}^t\}_t\right)\right\}.
\eeq

\begin{proposition}
\label{prop:varF}
For a channel $\mc{N}_{A'\rar B}$, any ensemble of quantum states $\{q(t),\psi_{AA'}^t\}_{t\in\mc{T}}$
and any $R \geq 0$, it holds that
\begin{align}
 &F\left(\mc{N}_{A' \rar B},R,\{q(t),\psi_{AA'}^t\}_t\right) \nb\\
=&\inf_{\{\tau_{AB}^t\}_t\in\mc{F}} \left\{\Big(R-\sum_{t \in \mc{T}} q(t)D(\tau_{AB}^t \| \psi_A^t \ox \tau_B^t)\Big)_+ +\sum_{t \in \mc{T}} q(t)D\big(\tau_{AB}^t \| \mc{N}(\psi_{AA'}^t)\big) \right\},
\end{align}
where $\mc{F}:=\big\{\{\tau_{AB}^{t}\}_{t\in\mc{T}}~|~\tau_{AB}^{t} \in \mc{S}_{\mc{N}(\psi_{AA'}^t)}(AB),~\forall t \in \mc{T}\big\}$, and $x_+:=\max\{0,x\}$
for a real number $x$.
\end{proposition}
\begin{proof}
Define the following sets
\begin{align}
\mc{O}^t:=&\big\{ \sigma_B ~|~\sigma_B\in\mc{S}(B),\
\supp (\sigma_B)\supseteq\supp(\mc{N}(\psi_{A'}^t)) \big\},\\
\mc{F}^t:=&\big\{\tau_{AB}~|~\tau_{AB}\in \mc{S}_{\mc{N}(\psi_{AA'}^t)}(AB) \big\}
\end{align}
for all $t\in\mc{T}$. The minimization in Eq.~(\ref{eq:mic1-logE}) can be restricted to those
$\sigma_B$ whose support contains the support of $\mc{N}(\psi_{A'}^t)$, because
otherwise $D_{\alpha}^\flat(\mc{N}(\psi_{AA'}^t) \| \psi^t_A \ox \sigma_B)=+\infty$.
Thus we have
\begin{align}
&F\left(\mc{N}_{A' \rar B},R,\{q(t),\psi_{AA'}^t\}_t\right) \nb\\
=&\sup_{\alpha>1} \frac{\alpha-1}{\alpha} \left \{ R-\sum_{t \in \mc{T}} q(t)\inf_{\sigma_B\in\mc{O}^t}
D_{\alpha}^{\flat}(\mc{N}_{A' \rar B}(\psi_{AA'}^t) \| \psi_A^t\ox \sigma_B) \right \} \nb\\
=&\sup_{\alpha>1} \sum_{t \in \mc{T}} q(t) \sup_{\sigma_B\in\mc{O}^t} \inf_{\tau_{AB}\in\mc{F}^t}
  \left\{\frac{\alpha-1}{\alpha} \Big(R-D(\tau_{AB} \| \psi^t_A \ox \sigma_B)\Big)
  +D\big(\tau_{AB} \| \mc{N}(\psi_{AA'}^t)\big) \right\}, \label{equ:mini}
\end{align}
where the last equality follows from the variational expression of $D_{\alpha}^{\flat}(\rho \| \sigma)$~(see Proposition \ref{prop:renyid} (\romannumeral3)). Let
\begin{align}
f_t(\sigma_B, \tau_{AB})
:=&\frac{\alpha-1}{\alpha}\Big(R-D(\tau_{AB} \| \psi^t_A \ox \sigma_B)\Big)
   +D(\tau_{AB} \| \mc{N}(\psi_{AA'}^t)) \nb \\
 =&\frac{\alpha-1}{\alpha}\Big(R+\tr\left[\tau_{AB}\log(\psi^t_A \ox \sigma_B)\right]\Big)
   -\tr\left[\tau_{AB}\log\mc{N}(\psi_{AA'}^t)\right]-\frac{1}{\alpha}H(\tau_{AB}),
\end{align}
where $H(\rho):=-\tr(\rho\log\rho)$ is the von Neumann entropy. The relative entropy is
convex with respect to both arguments, and the von Neumann entropy is concave. So,
we can easily verify that, for any $\alpha>1$, $f_t$ satisfies the following properties:
\begin{enumerate}[(i)]
  \item for fixed $\tau_{AB}$, $f_t(\cdot,\tau_{AB})$ is concave and continuous on $\mc{O}^t$, and $\mc{O}^t$ is convex;
  \item for fixed $\sigma_{B}$, $f_t(\sigma_{B},\cdot)$ is convex and continuous on $\mc{F}^t$, and $\mc{F}^t$ is a compact convex set.
\end{enumerate}
So, Sion's minimax theorem (see Lemma~\ref{lemma:minimax} in the Appendix) applies. This lets us proceed as
\begin{align}
&F\left(\mc{N}_{A' \rar B},R,\{q(t),\psi_{AA'}^t\}_t\right) \nb\\
=&\sup_{\alpha>1} \sum_{t \in \mc{T}} q(t) \inf_{\tau_{AB}\in\mc{F}^t} \sup_{\sigma_B\in\mc{O}^t} \left\{\frac{\alpha-1}{\alpha} \Big(R-D(\tau_{AB} \| \psi^t_A \ox \sigma_B)\Big)
+D\big(\tau_{AB} \| \mc{N}(\psi_{AA'}^t)\big) \right\} \nb\\
=&\sup_{\alpha>1} \sum_{t \in \mc{T}} q(t) \inf_{\tau_{AB}\in\mc{F}^t} \left\{\frac{\alpha-1}{\alpha} \Big(R-D(\tau_{AB} \| \psi^t_A \ox \tau_B)\Big)
+D\big(\tau_{AB} \| \mc{N}(\psi_{AA'}^t)\big) \right\} \nb\\
=&\sup_{\lambda \in (0,1)}  \inf_{\{\tau_{AB}^t\}_t\in\mc{F}} \left \{\lambda\Big(R-\sum_{t \in \mc{T}} q(t)D(\tau_{AB}^t \| \psi^t_A \ox \tau_B^t)\Big)+\sum_{t \in \mc{T}} q(t)D\big(\tau_{AB}^t \| \mc{N}(\psi_{AA'}^t)\big) \right \} \nb\\
=&\inf_{\{\tau_{AB}^t\}_t\in\mc{F}} \sup_{\lambda \in (0,1)} \left\{\lambda\Big(R-\sum_{t \in \mc{T}} q(t)D(\tau_{AB}^t \| \psi^t_A \ox \tau_B^t)\Big)+\sum_{t \in \mc{T}} q(t)D\big(\tau_{AB}^t \| \mc{N}(\psi_{AA'}^t)\big) \right\} \nb\\
=&\inf_{\{\tau_{AB}^t\}_t\in\mc{F}} \left\{\Big(R-\sum_{t \in \mc{T}} q(t)D(\tau_{AB}^t \| \psi^t_A \ox \tau_{B}^t)\Big)_+ + \sum_{t \in \mc{T}} q(t)D\big(\tau_{AB}^t \| \mc{N}(\psi_{AA'}^t)\big) \right\},
\end{align}
where for the fourth equality, we have used Sion's minimax theorem again. It can be verified similarly that the conditions for Sion's minimax theorem are satisfied. The convexity of the expression under optimization as a function of $\{\tau_{AB}^t\}_t$ is not obvious. To see this, we write it as
\begin{align}
 & \sum_{t \in \mc{T}} q(t)\left \{\lambda R-\lambda D(\tau_{AB}^t \| \psi^t_A \ox \tau_B^t)
   +D\big(\tau_{AB}^t \| \mc{N}(\psi_{AA'}^t)\big) \right\} \nb\\
=& \sum_{t \in \mc{T}} q(t)\left \{\lambda R-(1-\lambda)H(\tau_{AB}^t) -\lambda H(\tr_A\tau_{AB}^t)
   +\tr\left[\tau_{AB}^t\left(\lambda\log\psi^t_A-\log \mc{N}(\psi_{AA'}^t) \right)\right]\right\}.
\end{align}
The convexity follows from the fact that for any $t\in\mc{T}$, $-(1-\lambda)H(\tau_{AB}^t)$ and $-\lambda H(\tr_A\tau_{AB}^t)$ are convex as functions of $\tau_{AB}^t$, and $\tr\left[\tau_{AB}^t\left(\lambda\log\psi^t_A-\log \mc{N}(\psi_{AA'}^t) \right)\right]$
is linear.
\end{proof}

\medskip
To prove Proposition~\ref{prop:weak-euc}, we introduce
\begin{align}
\label{eq:F1}
F_1\big(\mc{N},R,\{q(t),\Psi_{AA'}^t\}_t\big)
\!:=&\!\inf_{\{\tau_{AB}^t\}_t \in \mc{F}_1}
\sum_{t \in \mc{T}} q(t)D\big(\tau_{AB}^t \| \mc{N}(\Psi_{AA'}^t)\big), \\
\label{eq:F2}
F_2\big(\mc{N},R,\{q(t),\Psi_{AA'}^t\}_t\big)
\!:= &\!\inf_{\{\tau_{AB}^t\}_t \in \mc{F}_2}
\!\!\left\{\!R\!-\!\!\sum_{t \in \mc{T}}\! q(t)D\big(\tau_{AB}^t \| \pi_A^t\! \ox\! \tau_{B}^t\big)
\!+\!\!\sum_{t \in \mc{T}}\! q(t)D\big(\tau_{AB}^t \| \mc{N}(\Psi_{AA'}^t)\!\big)\!\right\}
\end{align}
with
\begin{align}
\label{eq:mcF1}
\mc{F}_1 &:=\Big\{ \{\tau_{AB}^t\}_{t \in \mc{T}}~\big|~
(\forall t\in\mc{T})\ \tau_{AB}^t\in\mc{S}_{\mc{N}(\Psi_{AA'}^t) }(AB),\
\sum_{t \in \mc{T}} q(t)D(\tau_{AB}^t \| \pi_A^t \ox \tau_{B}^t) > R \Big\}, \\
\label{eq:mcF2}
\mc{F}_2 &:=\Big\{ \{\tau_{AB}^t\}_{t \in \mc{T}}~\big|~
(\forall t\in\mc{T})\ \tau_{AB}^t\in\mc{S}_{\mc{N}(\Psi_{AA'}^t) }(AB),\
\sum_{t \in \mc{T}} q(t)D(\tau_{AB}^t \| \pi_A^t \ox \tau_{B}^t) \leq R \Big\},
\end{align}
where $\pi_A^t$ is the maximally mixed state on $\mc{H}_A^t$. It is obvious that
\begin{equation}
\label{equ:F1F2}
F\big(\mc{N},R,\{q(t),\Psi_{AA'}^t\}_t\big)
=\min\Big\{F_1\big(\mc{N},R,\{q(t),\Psi_{AA'}^t\}_t\big),
F_2\big(\mc{N},R,\{q(t),\Psi_{AA'}^t\}_t\big)\Big\}.
\end{equation}
So, it suffices to show that $sc(\mc{N}_{A' \rar B},R)$ is upper bounded by both $F_1$
and $F_2$, given in Eq.~\eqref{eq:F1} and Eq.~\eqref{eq:F2}, respectively.

\begin{proofof}{Proposition~\ref{prop:weak-euc}}
Thanks to Proposition~\ref{prop:varF}, it follows from Lemma~\ref{lemma:use1} and
Lemma~\ref{lemma:use2} below.
\end{proofof}

\medskip
\begin{lemma}
\label{lemma:use1}
Let $\{q(t),\Psi_{AA'}^t\}_{t\in\mc{T}}$ be any ensemble of quantum states as specified
in Proposition~\ref{prop:weak-euc}. For the channel $\mc{N}_{A' \rar B}$ and any $R>0$
we have
\beq
sc(\mc{N}_{A' \rar B},R) \leq F_1\big(\mc{N}_{A'\rar B},R,\{q(t),\Psi_{AA'}^t\}_t\big).
\eeq
\end{lemma}
\begin{proof}
If $\mc{F}_1=\emptyset$, the statement is trivial because the right hand side is $+\infty$.
So we suppose that $\mc{F}_1\neq\emptyset$. The definition of $F_1$ in Eq.~\eqref{eq:F1} and Eq.~\eqref{eq:mcF1} implies that for
any $\delta>0$, there exists a set of states $\{\tau_{AB}^t\}_{t\in\mc{T}}$ such that $\tau_{AB}^t\in\mc{S}_{\mc{N}(\Psi_{AA'}^t)}(AB)$ and
\begin{align}
\sum_{t \in \mc{T}} q(t)D\big(\tau_{AB}^t \| \pi_A^t \ox \tau_B^t\big) &> R, \label{eq:use1-a}\\
\sum_{t \in \mc{T}} q(t)D\big(\tau_{AB}^t \| \mc{N}(\Psi_{AA'}^t)\big)
&\leq F_1\big(\mc{N}_{A'\rar B},R,\{q(t),\Psi_{AA'}^t\}_t\big)+\delta. \label{eq:use1-b}
\end{align}
We will employ the Heisenberg-Weyl operators. Defined on a $d$-dimensional Hilbert space
$\mc{H}$ with an orthonormal basis $\{\ket{x}\}_{x=0}^{d-1}$, they are a collection of
unitary operators
\beq
V_{y,z}
=\sum_{x=0}^{d-1}\mathrm{e}^{\frac{2\pi\mathrm{i}xz}{d}}
 \ket{(x+y)\!\!\!\mod d}\bra{x},
\eeq
where $y,z\in \{0,1,\ldots,d-1\}$. Let $\mc{V}^t$ be the set of Heisenberg-Weyl
operators on $\mc{H}_A^t$, defined with respect to the basis $\{\ket{a^t_x}\}_x$
for which $\ket{\Psi^t}_{AA'}=\frac{1}{\sqrt{|\mc{H}_A^t|}}
\sum_{x} \ket{a^t_x}_A \ox \ket{a^t_x}_{A'}$. This ensures
that for any $U^t\in\mc{V}^t$,
\beq\label{eq:uut}
U^t \ox \1 \ket{\Psi^t}_{AA'} = \1 \ox (U^t)^T \ket{\Psi^t}_{AA'},
\eeq
where $(U^t)^T$ is the transpose of $U^t$ and it acts on $\mc{H}_{A'}^t$. Let
\beq
\mc{U}:=\Big\{\bigoplus_{t\in\mc{T}}U^t~\big|~(\forall t)~U^t\in\mc{V}^t\Big\}
\eeq
be a set of unitary operators on $\mc{H}_A$. We consider the ensemble of quantum states
\beq\label{eq:ensemble}
\left\{\frac{1}{|\mc{U}|}, U_A \big(\sum_{t\in\mc{T}}q(t)\tau_{AB}^t\big) U_A^\dagger \right\}_{U_A\in\mc{U}}.
\eeq
Its Holevo information is evaluated as
\begin{align}
&\sum_{{U_A}\in\mc{U}}\frac{1}{|\mc{U}|}D\Big(U_A \big(\sum_{t\in\mc{T}}q(t)\tau_{AB}^t\big)
   U_A^\dagger \big\|\sum_{U_A\in\mc{U}}\frac{1}{|\mc{U}|} U_A
   \big(\sum_{t\in\mc{T}}q(t)\tau_{AB}^t\big) U_A^\dagger \Big) \nb\\
=&\sum_{{U_A}\in\mc{U}}\frac{1}{|\mc{U}|}D\Big(U_A \big(\sum_{t\in\mc{T}}q(t)\tau_{AB}^t\big)
   U_A^\dagger \big\| \sum_{t\in\mc{T}}q(t)\pi_{A}^t\ox\tau_{B}^t \Big) \nb\\
=&D\Big(\sum_{t\in\mc{T}}q(t)\tau_{AB}^t \big\|\sum_{t\in\mc{T}}q(t)\pi_{A}^t\ox\tau_{B}^t \Big) \nb\\
=&\sum_{t \in \mc{T}} q(t)D(\tau_{AB}^t \| \pi_A^t \ox \tau_B^t).
\end{align}
So, for the ensemble given in Eq.~\eqref{eq:ensemble} and $R$ satisfying Eq.~\eqref{eq:use1-a}, we are able to apply Lemma~\ref{lem:HSW}. Thus for $n\in\mathbb{N}$ one can construct a set of signal states
\beq
\mc{O}_n\equiv\Big\{\mc{E}_{{A}^n}^m
\big(\big(\sum_{t\in\mc{T}}q(t)\tau_{AB}^t\big)^{\ox n} \big)\Big\}_{m\in\mc{M}_n}
\eeq
encoding messages $m\in\mc{M}_n$ of size $|\mc{M}_n|=\lfloor2^{nR}\rfloor$, such that there exists a decoding measurement \beq\label{eq:dcoding}
\mc{D}_n\equiv\big\{\Lambda_{A^nB^n}^m\big\}_{m \in \mc{M}_n}
\eeq
with success probability
\begin{align}
\widetilde{P_s}((\mc{O}_n,\mc{D}_n))&:= \frac{1}{|\mc{M}_n|}\sum_{m \in \mc{M}_n}
\tr \mc{E}_{{A}^n}^m\big(\big(\sum_{t\in\mc{T}}q(t)\tau_{AB}^t\big)^{\ox n} \big)
\Lambda_{A^nB^n}^m \nb\\
&\rar 1, \quad\text{as } n\rar\infty. \label{eq:sdense}
\end{align}
Moreover, $\mc{E}_{{A}^n}^m(\cdot)= U_m(\cdot)U_m^\dagger $ is a unitary operation on $A^n$ of the form
\begin{equation}
\label{eq:form1}
U_m=\Big(\bigoplus_{t \in\mc{T}}U_1^t(m)\Big) \ox\cdots\ox \Big(\bigoplus_{t \in\mc{T}} U_n^t(m)\Big),
\end{equation}
where $U_i^t(m)\in\mc{V}^t$ is a Heisenberg-Weyl operator on $\mc{H}_A^t$.

Now, given the message set $\mc{M}_n$, we design a code $\mc{C}_n$ for
$\mc{N}_{A' \rar B}^{\ox n}$ as follows.
\begin{enumerate}[1.]
  \item Alice and Bob share the entangled state
      $\rho_{A^n{A'}^n}=\big(\sum_{t \in \mc{T}} q(t) \Psi_{AA'}^t\big)^{\ox n}$.
      Alice holds the ${A'}^n$ part, and Bob holds the $A^n$ part.
  \item Alice encodes the classical message $m\in\mc{M}_n$ by applying the map $({\mc{E}}_{{A'}^n}^m)^T$ to the ${A'}^n$ part of the state $\rho_{A^n{A'}^n}$. Here $({\mc{E}}_{{A'}^n}^m)^T(\cdot)= (U_m)^T(\cdot)((U_m)^T)^\dagger $ and
      \begin{equation}
        \label{eq:form2}
        (U_m)^T=\Big(\bigoplus_{t \in\mc{T}}(U_1^t(m))^T\Big) \ox\cdots\ox
        \Big(\bigoplus_{t \in\mc{T}} (U_n^t(m))^T\Big)
      \end{equation}
      is the transpose of $U_m$. Then she sends the ${A'}^n$ part to Bob through the channel
      $\mc{N}_{A' \rar B}^{\ox n}$.
  \item Bob performs the decoding measurement $\{\Lambda_{A^nB^n}^m\}_{m \in \mc{M}_n}$ of Eq.~\eqref{eq:dcoding}, aiming to recover the message $m$.
\end{enumerate}
The success probability of the code $\mc{C}_n$ can be evaluated as
\begin{align}
&P_s(\mc{N}_{A' \rar B}^{\ox n},\mc{C}_n) \nb\\
=&\frac{1}{|\mc{M}_n|}\sum_{m \in \mc{M}_n}
\tr\mc{N}_{A'\rar B}^{\ox n}\circ({\mc{E}}_{{A'}^n}^m)^T(\rho_{A^n{A'}^n})\Lambda_{A^nB^n}^m \nb\\
=&\frac{1}{|\mc{M}_n|}\sum_{m \in \mc{M}_n}
\tr\big(\mc{E}_{{A}^n}^m\ox\mc{N}_{A'\rar B}^{\ox n}\big)(\rho_{A^n{A'}^n})\Lambda_{A^nB^n}^m \nb\\
=&\frac{1}{|\mc{M}_n|}\sum_{m \in \mc{M}_n}
\tr \mc{E}_{{A}^n}^m \big(\big(\sum_{t \in \mc{T}} q(t) \mc{N}(\Psi_{AA'}^t)\big)^{\ox n}\big) \Lambda_{A^nB^n}^m, \label{eq:sN}
\end{align}
where the second equality is by the relation of Eq.~\eqref{eq:uut}.

To proceed, we set
\beq\label{eq:a}
a=\sum_{t \in \mc{T}} q(t)D\big(\tau_{AB}^t \| \mc{N}(\Psi_{AA'}^t)\big)+\delta.
\eeq
Then Eq.~(\ref{eq:sdense}) and Eq.~(\ref{eq:sN}) together lead to
\begin{align}
&\widetilde{P_s}((\mc{O}_n,\mc{D}_n))-2^{na}P_s(\mc{N}_{A' \rar B}^{\ox n},\mc{C}_n) \nb\\
=&\frac{1}{|\mc{M}_n|}\sum_{m \in \mc{M}_n} \tr \mc{E}_{{A}^n}^m
  \Big(\big(\sum_{t \in \mc{T}} q(t)\tau_{AB}^t\big)^{\ox n}-2^{na}
  \big(\sum_{t\in\mc{T}}q(t)\mc{N}(\Psi_{AA'}^t)\big)^{\ox n}\Big)\Lambda_{A^nB^n}^m \nb\\
=&\frac{1}{|\mc{M}_n|}\sum_{m \in \mc{M}_n} \tr \Big(\big(\sum_{t \in \mc{T}} q(t)
  \tau_{AB}^t\big)^{\ox n}-2^{na}\big(\sum_{t\in\mc{T}}q(t)\mc{N}(\Psi_{AA'}^t)\big)^{\ox n}\Big)
  \big(U_m^\dagger \Lambda_{A^nB^n}^mU_m\big) \nb\\
\leq& \tr\Big(\big(\sum_{t\in\mc{T}}q(t)\tau_{AB}^t\big)^{\ox n}-2^{na}
  \big(\sum_{t\in\mc{T}}q(t)\mc{N}(\Psi_{AA'}^t)\big)^{\ox n}\Big)_+, \label{eq:compare}
\end{align}
where $X_+:=(|X|+X)/2$ is the positive part of a self-adjoint operator $X$.
It has been proved in~\cite{OgawaNagaoka2000strong} (cf. Theorem 1 and discussions in Section~\uppercase\expandafter{\romannumeral3} therein) that when $r>D(\rho \| \sigma)$,
\beq
\lim_{n \rar \infty} \tr (\rho^{\ox n} - 2^{nr} \sigma^{\ox n})_+=0.
\eeq
Thus, when $n\rar\infty$, the limit of the last line of Eq.~\eqref{eq:compare} is $0$.
Therefore, we combine Eq.~\eqref{eq:sdense} and Eq.~\eqref{eq:compare} to get
\beq\label{eq:sNbound}
\liminf_{n\rar\infty} 2^{na}P_s(\mc{N}_{A' \rar B}^{\ox n},\mc{C}_n) \geq 1.
\eeq
Since the code $\mc{C}_n$ has size $|\mc{C}_n|= |\mc{M}_n|=2^{nR}$, by the definition of
$sc(\mc{N}_{A' \rar B},R)$, we have
\begin{align}
     &sc(\mc{N}_{A' \rar B},R) \nb\\
\leq &-\liminf_{n\rar\infty} \frac{1}{n}\log P_s(\mc{N}_{A' \rar B}^{\ox n},\mc{C}_n) \nb\\
\leq &\sum_{t \in \mc{T}} q(t)D\big(\tau_{AB}^t \| \mc{N}(\Psi_{AA'}^t)\big)+\delta \nb\\
\leq &F_1\big(\mc{N}_{A'\rar B},R,\{q(t),\Psi_{AA'}^t\}_t\big)+2\delta.
\end{align}
where the second inequality is by Eq.~\eqref{eq:a} and Eq.~\eqref{eq:sNbound}, and the last
inequality is by Eq.~\eqref{eq:use1-b}. At last, because $\delta>0$ is arbitrary, we are done.
\end{proof}

\medskip
\begin{lemma}
\label{lemma:use2}
Let $\{q(t),\Psi_{AA'}^t\}_{t\in\mc{T}}$ be any ensemble of quantum states as specified
in Proposition~\ref{prop:weak-euc}. For the channel $\mc{N}_{A' \rar B}$ and any $R>0$
we have
\beq
sc(\mc{N}_{A' \rar B},R) \leq F_2\big(\mc{N}_{A'\rar B},R,\{q(t),\Psi_{AA'}^t\}_t\big).
\eeq
\end{lemma}
\begin{proof}
According to the definition of $F_2$ in Eq.~\eqref{eq:F2} and Eq.~\eqref{eq:mcF2},
there exists a set of states $\{\tau_{AB}^t\}_{t\in\mc{T}}$ such that $\tau_{AB}^t\in\mc{S}_{\mc{N}(\Psi_{AA'}^t)}(AB)$ and
\begin{align}
&\sum_{t \in \mc{T}} q(t)D\big(\tau_{AB}^t \| \pi_A^t \ox \tau_B^t\big) \leq R, \\
&R-\sum_{t \in \mc{T}}q(t)D\big(\tau_{AB}^t \| \pi_A^t \ox \tau_B^t\big)
+\sum_{t \in \mc{T}} q(t)D\big(\tau_{AB}^t \| \mc{N}(\Psi_{AA'}^t)\big)
=F_2\big(\mc{N}_{A'\rar B},R,\{q(t),\Psi_{AA'}^t\}_t\big).
\end{align}
Let $R':=\sum_{t \in \mc{T}} q(t)D\big(\tau_{AB}^t \|\pi_A^t \ox \tau_B^t\big)-\delta$
with an arbitrary $\delta>0$. Following the proof of Lemma~\ref{lemma:use1}, we easily
see that there exists a sequence of codes
\beq
\mc{C}'_n=
\Big(\big\{{\mc{E}'}_{{A'}^n}^m\big\}_m,\big\{{\Lambda'}_{A^nB^n}^m\big\}_m,
\big(\sum_{t \in \mc{T}}q(t)\Psi_{AA'}^t\big)^{\ox n}\Big),\quad n\in \mb{N}
\eeq
with $|\mc{C}'_n|=\lfloor2^{nR'}\rfloor$ such that
\begin{equation}
\label{equ:R1}
-\liminf_{n\rar\infty} \frac{1}{n}\log P_s(\mc{N}_{A' \rar B}^{\ox n},\mc{C}'_n)
\leq \sum_{t \in \mc{T}} q(t)D\big(\tau_{AB}^t \| \mc{N}(\Psi_{AA'}^t)\big).
\end{equation}
We point out that in the code $\mc{C}'_n$, the sender holds the $A'^n$ system and the
receiver holds the $A^n$ system of the entangled state
$(\sum_{t \in \mc{T}}q(t)\Psi_{AA'}^t)^{\ox n}$, and the encoding operation
${\mc{E}'}_{{A'}^n}^m$ is a unitary map. Now we extend $\mc{C}'_n$ to construct a new
code $\mc{C}_n$ with size $|\mc{C}_n|=\lfloor2^{nR}\rfloor$, using the simple strategy:
\begin{enumerate}[(i)]
  \item for $m \in \left\{1,2,\ldots,|\mc{C}'_n|\right\}$, we set
        $\mc{E}_{{A'}^n}^m={\mc{E}'}_{{A'}^n}^m$ and $\Lambda_{A^nB^n}^m={\Lambda'}_{A^nB^n}^m$;
  \item for $m \notin \left\{1,2,\ldots,|\mc{C}'_n|\right\}$, we set
        $\mc{E}_{{A'}^n}^m=\mc{I}_{{A'}^n}$ as the identity map and $\Lambda_{A^nB^n}^m=0$.
\end{enumerate}
The success probability of the code
\beq
\mc{C}_n=
\Big(\big\{{\mc{E}}_{{A'}^n}^m\big\}_m,\big\{{\Lambda}_{A^nB^n}^m\big\}_m,
\big(\sum_{t \in \mc{T}}q(t)\Psi_{AA'}^t\big)^{\ox n}\Big)
\eeq
is calculated as
\begin{align}
   &P_s(\mc{N}_{A'\rar B}^{\ox n},\mc{C}_n) \nb\\
 = &\sum_{m=1}^{|\mc{C}_n|}\frac{1}{|\mc{C}_n|} \tr\Big(\mc{N}^{\ox n}\circ\mc{E}_{{A'}^n }^m
\big(\big(\sum_t q(t) \Psi_{AA'}^t\big)^{\ox n}\big)\Big) \Lambda_{A^nB^n}^m \nb\\
 = &\frac{|\mc{C}'_n|}{|\mc{C}_n|}\sum_{m=1}^{|\mc{C}'_n|}\frac{1}{|\mc{C}'_n|}
    \tr\Big(\mc{N}^{\ox n}\circ{\mc{E}'}_{{A'}^n }^m\big(\big(\sum_t q(t)
    \Psi_{AA'}^t\big)^{\ox n}\big)\Big) {\Lambda'}_{A^nB^n}^m \nb\\
 = &\frac{|\mc{C}'_n|}{|\mc{C}_n|}P_s(\mc{N}_{A' \rar B}^{\ox n},\mc{C}'_n). \label{equ:R1R}
\end{align}
Thus we have
\begin{align}
sc(\mc{N}_{A' \rar B},R)
&\leq -\liminf_{n\rar\infty} \frac{1}{n}\log P_s(\mc{N}_{A' \rar B}^{\ox n},\mc{C}_n) \nb\\
&\leq R-R'+\sum_{t \in \mc{T}} q(t)D\big(\tau_{AB}^t \| \mc{N}(\Psi_{AA'}^t)\big)  \nb\\
& =   F_2\big(\mc{N}_{A'\rar B},R,\{q(t),\Psi_{AA'}^t\}_t\big)+\delta,
\end{align}
where the second inequality is by Eq.~(\ref{equ:R1}) and Eq.~(\ref{equ:R1R}), as well as
the choice of $|\mc{C}'_n|$ and $|\mc{C}_n|$. At last, because $\delta>0$ is arbitrary,
the proof is complete.
\end{proof}

\bigskip
The bound in Proposition~\ref{prop:weak-euc} depends on the quantum state set
$\{\Psi_{AA'}^t\}_t$ as well as the probability distribution $q$ over it. We optimize it
over all probability distribution $q$ to obtain the following result.
\begin{theorem}
\label{thm:strong-euc}
Let $\mc{H}_A\cong\mc{H}_{A'}$, and let $\mc{H}_A=\oplus_{t \in \mc{T}} \mc{H}_A^t$ and
$\mc{H}_{A'}=\oplus_{t \in \mc{T}} \mc{H}_{A'}^t$ be decompositions of $\mc{H}_A$ and
$\mc{H}_{A'}$ into orthogonal subspaces with $\mc{H}_A^t\cong\mc{H}_{A'}^t$. Let
$\Psi_{AA'}^t$ be the maximally entangled state on $\mc{H}_A^t \ox \mc{H}_{A'}^t$.
Then for the channel $\mc{N}_{A' \rar B}$ and any $R>0$, the strong converse exponent
for entanglement-assisted classical communication satisfies
\beq
sc(\mc{N}_{A'\rar B},R)
\leq \sup_{\alpha>1} \frac{\alpha-1}{\alpha}
\left\{ R-\max_{t\in\mc{T}}I_{\alpha}^{\flat}\left(\mc{N}_{A'\rar B},\Psi_{AA'}^t\right)\right\},
\eeq
where $I_{\alpha}^{\flat}\left(\mc{N}_{A'\rar B},\Psi_{AA'}^t\right)$ is defined in Eq.~\eqref{eq:mic1-logE-s}.
\end{theorem}

\begin{proof}
Proposition~\ref{prop:weak-euc} holds for any ensemble $\{q(t),\Psi_{AA'}^t\}_{t\in\mc{T}}$.
So, we can take the infimum in Eq.~\eqref{eq:weak-euc},
over all possible $q$ in the probability simplex $\mc{Q}(\mc{T})$.
Therefore,
\begin{align}
sc(\mc{N}_{A' \rar B},R)
&\leq \inf_{q\in\mc{Q}(\mc{T})} \sup_{\alpha>1} \frac{\alpha-1}{\alpha}
      \left\{R-I_{\alpha}^{\flat}\left(\mc{N}_{A'\rar B},\{q(t),\Psi_{AA'}^t\}_t\right)\right\} \nb\\
& =   \inf_{q\in\mc{Q}(\mc{T})} \sup_{\lambda\in(0,1)} G(\lambda,q), \label{eq:strong-euc-1}
\end{align}
where
\begin{equation}
  \label{eq:G1}
G(\lambda,q)
:=\lambda\left\{R-\sum_{t \in \mc{T}} q(t)\min_{\sigma_B} D_{\frac{1}{1-\lambda}}^{\flat}\left(\mc{N}(\Psi_{AA'}^t)\|\pi_A^t\ox\sigma_B\right)\right\},
\end{equation}
and for the equality we use the definition of $I_{\alpha}^{\flat}\left(\mc{N}_{A'\rar B},\{q(t),\Psi_{AA'}^t\}_t\right)$ and set $\lambda=\frac{\alpha-1}{\alpha}$. Following
the steps in the proof of Proposition~\ref{prop:varF}, we can also write $G(\lambda,q)$ as
\begin{equation}
  \label{eq:G2}
G(\lambda,q)
=\inf_{\{\tau_{AB}^t\}_t\in\mc{F}}
 \left\{\lambda \Big(R-\sum_{t \in \mc{T}} q(t)D(\tau_{AB}^t \| \pi^t_A \ox \tau_B^t)\Big)
 +\sum_{t \in \mc{T}} q(t)D\big(\tau_{AB}^t \| \mc{N}(\Psi_{AA'}^t)\big) \right\},
\end{equation}
where $\mc{F}:=\{\{\tau_{AB}^{t}\}_{t\in\mc{T}}~|~\tau_{AB}^{t}\in\mc{S}_{\mc{N}(\Psi_{AA'}^t)}(AB),
~\forall t \in \mc{T}\}$. We claim that
\begin{enumerate}[(i)]
  \item for fixed $\lambda$, the function $q\mapsto G(\lambda,q)$ is linear and continuous,
  on the compact and convex set $\mc{Q}(\mc{T})$; and
  \item for fixed $q$, the function $\lambda\mapsto G(\lambda,q)$ is concave and upper
  semi-continuous, on the interval $(0,1)$.
\end{enumerate}
Claim (\romannumeral1) is obviously seen from Eq.~\eqref{eq:G1}. For claim (\romannumeral2), we can easily check that $x\mapsto\inf_y f(x,y)$ is concave and upper semi-continuous if the function $f(x,y)$ is linear with $x$ and continuous with both arguments, and then we apply this observation to Eq.~\eqref{eq:G2}.
Now we can invoke Sion's minimax theorem to obtain
\begin{align}
 &\inf_{q\in\mc{Q}(\mc{T})} \sup_{\lambda \in (0,1)} G(\lambda,q) \nb\\
=&\sup_{\lambda \in (0,1)} \inf_{q\in\mc{Q}(\mc{T})} G(\lambda,q) \nb\\
=&\sup_{\lambda \in (0,1)} \lambda\left\{R-\max_{t\in\mc{T}}\min_{\sigma_B} D_{\frac{1}{1-\lambda}}^{\flat}\left(\mc{N}(\Psi_{AA'}^t)\|\pi_A^t\ox\sigma_B\right)\right\} \nb\\
=&\sup_{\alpha>1} \frac{\alpha-1}{\alpha}
\left\{ R-\max_{t\in\mc{T}}I_{\alpha}^{\flat}\left(\mc{N}_{A'\rar B},\Psi_{AA'}^t\right)\right\},
\end{align}
and we are done.
\end{proof}

\medskip
\begin{remark}\label{rk:codes}
To achieve the bound of Theorem~\ref{thm:strong-euc}, we can employ a sequence of codes
\begin{equation}
\left\{\Big(\big\{{\mc{E}}_{{A'}^n}^m\big\}_m,\big\{{\Lambda}_{A^nB^n}^m\big\}_m,
\big(\sum_{t \in \mc{T}}\bar{q}(t)\Psi_{AA'}^t\big)^{\ox n}\Big)\right\}_n,
\end{equation}
where the encoding operation ${\mc{E}}_{{A'}^n}^m$ is a unitary map, and the probability
distribution $\bar{q}(t)$ in the shared entangled state is an optimizer of
Eq.~\eqref{eq:strong-euc-1}. This can be seen from the construction of codes in
the proofs of Lemma~\ref{lemma:use1} and Lemma~\ref{lemma:use2}. Note that the
infima in Eq.~(\ref{eq:strong-euc-1}) can be replaced by minima (cf. Sion's
minimax theorem in Lemma~\ref{lemma:minimax}).
\end{remark}

\section{Achieving the Strong Converse Exponent}
  \label{sec:proof-final}
In this section, we complete the proof of the achievability of the strong converse
exponent for entanglement-assisted classical communication. This is based
on the result obtained in the previous section, namely, Theorem~\ref{thm:strong-euc}.
\begin{theorem}
\label{thm:main-upper}
Let $\mc{N}_{A' \rar B}$ be a quantum channel. For any transmission rate $R>0$, the
strong converse exponent for entanglement-assisted classical communication satisfies
\begin{equation}
\label{eq:scexponent-upper}
sc(\mc{N}_{A' \rar B},R) \leq \sup_{\alpha>1}
\frac{\alpha-1}{\alpha}\big(R-I_{\alpha}^*(\mc{N}_{A' \rar B})\big).
\end{equation}
\end{theorem}

\begin{proof}
We fix $m\in\mb{N}$, and consider the channel $\mc{N}^{(m)}_{A'^m\rar B^m}:=
\mc{P}_{\sigma_{B^m}^u}\circ\mc{N}_{A' \rar B}^{\ox m}$, where $\sigma_{B^m}^u$
is the universal symmetric state described in Lemma~\ref{lem:u-sym}. Let $A$ be such
that $\mc{H}_A\cong\mc{H}_{A'}$. We further fix an arbitrary pure state ${\psi}_{AA'}$
on $\mc{H}_A\ox\mc{H}_{A'}$. The tensor product state ${\psi}_{AA'}^{\ox m}$ can be
expressed in the form of Eq.~(\ref{equ:typedecompose}), i.e.,
\begin{equation}
\ket{\psi}_{AA'}^{\ox m}=\sum_{t \in \mc{T}_m} \sqrt{p^m(t)}\ket{\Psi^t}_{A^mA'^m},
\end{equation}
where $|\mc{T}_m|\leq (m+1)^{|A|}$ and the set of states $\{\Psi^t_{A^mA'^m}\}_t$
satisfies the condition of Theorem~\ref{thm:strong-euc} for the channel
$\mc{N}^{(m)}_{A'^m\rar B^m}$. Let $\pi_{A^m}^t=\Psi_{A^m}^t$ be the maximally mixed
state on the subspace $\mc{H}_{A^m}^t:=\supp(\Psi_{A^m}^t)$. By Theorem~\ref{thm:strong-euc},
we know that there exists a sequence of codes $\{\mc{C}_k^{(m)}\}_{k \in \mb{N}}$,
such that
\begin{equation}
\label{equ:limitrelation}
\liminf_{k\rar\infty} \frac{1}{k} \log |\mc{C}_k^{(m)}| \geq mR
\end{equation}
and
\begin{align}
     &-\liminf_{k\rar\infty} \frac{1}{k} \log P_s\left((\mc{N}^{(m)}_{A'^m\rar B^m})
                             ^{\ox k},\mc{C}_k^{(m)}\right) \nb\\
\leq & \sup_{\alpha>1} \frac{\alpha-1}{\alpha}\left\{mR-\max_{t\in\mc{T}_m}I_{\alpha}^{\flat}
       \left(\mc{N}^{(m)}_{A'^m\rar B^m},\Psi^t_{A^mA'^m}\right)\right\} \nb\\
  =  & \sup_{\alpha>1} \frac{\alpha-1}{\alpha}\left\{mR-\max_{t\in\mc{T}_m}
       \min_{\sigma_{B^m}\in\mc{S}(B^m)} D_{\alpha}^{\flat}\left(\mc{N}^{(m)}(\Psi^t_{A^mA'^m})
       \big\|\pi^t_{A^m}\ox \sigma_{B^m}\right) \right\}. \label{equ:expom}
\end{align}
As explained in Remark~\ref{rk:codes}, the code $\mc{C}_k^{(m)}$ can be written as
\beq\label{eq:code-km}
\mc{C}_k^{(m)}=\left(\big\{\mc{E}_{A'^{mk}}^i\big\}_i, \big\{\Lambda_{A^{mk}B^{mk}}^i\big\}_i,
\rho_{A^{mk}A'^{mk}}\right),
\eeq
where $\mc{E}_{A'^{mk}}^i$ is a unitary operation on $A'^{mk}$, and the state
$\rho_{A^{mk}A'^{mk}}$ shared by the sender (with $A'^{mk}$) and the
receiver (with $A^{mk}$) is a $k$-fold tensor product of a mixture of the states
$\{\Psi^t_{A^mA'^m}\}_t$. With this, the success probability of $\mc{C}_k^{(m)}$ can
be written as
\begin{align}
 &P_s\left(\big(\mc{N}^{(m)}_{A'^m\rar B^m}\big)^{\ox k},\mc{C}_k^{(m)}\right) \nb\\
=&\frac{1}{|\mc{C}_k^{(m)}|} \sum_{i=1}^{|\mc{C}_k^{(m)}|}\tr\big((\mc{N}^{(m)})^{\ox k}
  \circ \mc{E}_{{A'}^{mk}}^i (\rho_{A^{mk}A'^{mk}})\big)\Lambda_{A^{mk}B^{mk}}^i \nb\\
=&\frac{1}{|\mc{C}_k^{(m)}|} \sum_{i=1}^{|\mc{C}_k^{(m)}|}
  \tr\big(\mc{P}_{\sigma_{B^m}^u}^{\ox k}\circ\mc{N}^{\ox mk}\circ
  \mc{E}_{{A'}^{mk}}^i (\rho_{A^{mk}A'^{mk}})\big)\Lambda_{A^{mk}B^{mk}}^i \nb\\
=&\frac{1}{|\mc{C}_k^{(m)}|} \sum_{i=1}^{|\mc{C}_k^{(m)}|}\tr\big(
  \mc{N}^{\ox mk}\circ \mc{E}_{{A'}^{mk}}^i (\rho_{A^{mk}A'^{mk}})\big)
  {\mc{P}_{\sigma_{B^m}^u}^{\ox k}}(\Lambda_{A^{mk}B^{mk}}^i). \label{eq:Ps-km}
\end{align}

Now, for any integer $n$, we construct for the channel $\mc{N}_{A' \rar B}^{\ox n}$ a code $\mc{C}_n$, adapted from $\{\mc{C}_k^{(m)}\}_k$. Write $n=mk+l$ with $0\leq l<m$. We divide the channel
$\mc{N}_{A' \rar B}^{\ox n}=\mc{N}_{A' \rar B}^{\ox mk}\ox\mc{N}_{A' \rar B}^{\ox l}$ into the first
$mk$ copies and the last $l$ copies. The code $\mc{C}_n$ is performed effectively on the first $mk$
copies of $\mc{N}_{A'\rar B}$, specified by
\beq
\left(\left\{\mc{E}_{A'^{mk}}^i\right\}_i,
\big\{\mc{P}_{\sigma_{B^m}^u}^{\ox k}\big(\Lambda_{A^{mk}B^{mk}}^i\big)\big\}_i,
\rho_{A^{mk}A'^{mk}}\right).
\eeq
That is, we use the same entangled state and employ the same encoding maps with $\mc{C}_k^{(m)}$,
and we replace the decoding measurement operator by $\mc{P}_{\sigma_{B^m}^u}^{\ox k}(\Lambda_{A^{mk}B^{mk}}^i)$ (cf. Eq.~\eqref{eq:code-km}). On the last $l$ copies,
we input an arbitrary state and we do not touch the output in the decoding (equivalently,
all of the decoding measurement operators are tensored with an identity operator on this part).
The codes $\mc{C}_n$ and $\mc{C}_k^{(m)}$ have the same size. So it follows from Eq.~(\ref{equ:limitrelation}) that
\begin{equation}
\label{equ:limitrel}
\liminf_{n\rar\infty} \frac{1}{n} \log|\mc{C}_n|=\frac{1}{m}\liminf_{k\rar\infty} \frac{1}{k} \log|\mc{C}_k^{(m)}| \geq R.
\end{equation}
On the other hand, it is obvious that the success probability of $\mc{C}_n$ is given by Eq.~\eqref{eq:Ps-km} as well. So,
\beq\label{eq:Ps-equal}
P_s\left(\mc{N}_{A' \rar B}^{\ox n}, \mc{C}_n\right)
=P_s\left((\mc{N}^{(m)}_{A'^m\rar B^m})^{\ox k},\mc{C}_k^{(m)}\right).
\eeq
By the definition of $sc(\mc{N}_{A' \rar B},R)$, we combine Eq.~(\ref{equ:expom}), Eq.~(\ref{eq:Ps-equal}) and Eq.~(\ref{equ:limitrel}) to get
\begin{align}
&sc(\mc{N}_{A' \rar B},R) \nb\\
\leq&-\liminf_{n\rar\infty}\frac{1}{n}\log P_s\left(\mc{N}_{A'\rar B}^{\ox n},\mc{C}_n\right) \nb\\
=&-\frac{1}{m}\liminf_{k\rar\infty} \frac{1}{k}
  \log P_s\left((\mc{N}^{(m)}_{A'^m\rar B^m})^{\ox k},\mc{C}_k^{(m)}\right) \nb\\
\leq &  \sup_{\alpha>1} \frac{\alpha-1}{\alpha}
\left\{R-\frac{1}{m}\max_{t\in\mc{T}_m}\min_{\sigma_{B^m}\in\mc{S}(B^m)} D_{\alpha}^{\flat}
\left(\mc{P}_{\sigma_{B^m}^u}\circ\mc{N}^{\ox m}\left(\Psi^t_{A^mA'^m}\right) \big\|
\pi^t_{A^m}\ox \sigma_{B^m}\right) \right\}. \label{eq:zero}
\end{align}

Next, we further bound $sc(\mc{N}_{A' \rar B},R)$ based on Eq.~(\ref{eq:zero}). We have
\begin{align}
     &\frac{1}{m}\max_{t\in\mc{T}_m}\min_{\sigma_{B^m}\in\mc{S}(B^m)} D_{\alpha}^{\flat}
      \left(\mc{P}_{\sigma_{B^m}^u}\circ\mc{N}^{\ox m}\left(\Psi^t_{A^mA'^m}\right)
      \big\|\pi^t_{A^m}\ox \sigma_{B^m}\right) \nb\\
\stackrel{(a)}{=} &\frac{1}{m}\max_{t\in\mc{T}_m}\min_{\sigma_{B^m}\in\mc{S}_{\rm{sym}}(B^m)}
      D_{\alpha}^{\flat}\left(\mc{P}_{\sigma_{B^m}^u}\circ\mc{N}^{\ox m}
      \left(\Psi^t_{A^mA'^m}\right)\big\| \pi^t_{A^m}\ox \sigma_{B^m}\right)  \nb\\
\stackrel{(b)}{\geq} &\frac{1}{m}\max_{t\in\mc{T}_m} D_{\alpha}^{\flat}
      \left(\mc{P}_{\sigma_{B^m}^u}\circ\mc{N}^{\ox m}\left(\Psi^t_{A^mA'^m}\right)
      \big\| \pi^t_{A^m}\ox \sigma_{B^m}^u\right)-\frac{\log v_{m,|B|}}{m}  \nb\\
\stackrel{(c)}{=}  &\frac{1}{m}\max_{t\in\mc{T}_m} D_{\alpha}^{*}
      \left(\mc{P}_{\sigma_{B^m}^u}\circ\mc{N}^{\ox m}\left(\Psi^t_{A^mA'^m}\right)
      \big\| \pi^t_{A^m}\ox \sigma_{B^m}^u\right)-\frac{\log v_{m,|B|}}{m} \nb\\
\stackrel{(d)}{=}  &\frac{1}{m}\sup_{q}D_{\alpha}^{*}\Big(\mc{P}_{\sigma_{B^m}^u}\circ\mc{N}^{\ox m}
      \big(\sum_{t \in \mc{T}_m}q(t)\Psi^t_{A^mA'^m}\big)\big\|
      \sum_{t\in\mc{T}_m}q(t)\pi^t_{A^m}\ox\sigma_{B^m}^u\Big)-\frac{\log v_{m,|B|}}{m} \nb\\
\geq &\frac{1}{m} D_{\alpha}^{*}\Big(\mc{P}_{\sigma_{B^m}^u}\circ\mc{N}^{\ox m}
      \big(\sum_{t \in \mc{T}_m}p^m(t)\Psi^t_{A^mA'^m}\big)\big\| \sum_{t \in \mc{T}_m}
      p^m(t) \pi^t_{A^m}\ox \sigma_{B^m}^u\Big)-\frac{\log v_{m,|B|}}{m}, \label{eq:first}
\end{align}
where $(a)$ is essentially due to Proposition~\ref{prop:renyid}~(\romannumeral5)
and is shown in Lemma~\ref{lem:symmetry}, $(b)$ is by Lemma~\ref{lem:u-sym}
and Proposition \ref{prop:renyid}~(\romannumeral2), $(c)$ is because $\mc{P}_{\sigma_{B^m}^u}\circ\mc{N}^{\ox m}(\Psi^t_{A^mA'^m})$ and $ \pi^t_{A^m} \ox
\sigma_{B^m}^u$ commute, and $(d)$ can be easily verified from the definition
of $D_{\alpha}^{*}$. To go ahead, we introduce
\beq
\mc{P}_{\mc{T}_m}(\cdot)=\sum_{t \in \mc{T}_m} \Pi_t (\cdot) \Pi_t
\eeq
with $\Pi_t$ being the projection onto the subspace $\mc{H}_{A^m}^t$, and further bound
Eq.~(\ref{eq:first}) as follows.
\begin{align}
     &\frac{1}{m} D_{\alpha}^* \Big(\mc{P}_{\sigma_{B^m}^u}\circ\mc{N}^{\ox m}
      \big(\sum_{t \in \mc{T}_m } p^m(t)\Psi^t_{A^mA'^m}\big)\big\|\sum_{t \in \mc{T}_m }
      p^m(t)\pi^t_{A^m}\ox\sigma_{B^m}^u\Big)-\frac{\log v_{m,|B|}}{m}  \nb\\
 =   &\frac{1}{m} D_{\alpha}^* \left(\big(\mc{P}_{\mc{T}_m}\ox\mc{P}_{\sigma_{B^m}^u}\big)
      \big(\mc{N}^{\ox m}(\psi_{AA'}^{\ox m})\big) \big\|
      \psi_{A}^{\ox m} \ox \sigma_{B^m}^u\right)-\frac{\log v_{m,|B|}}{m} \nb\\
\stackrel{(a)}{\geq} &\frac{1}{m} D_{\alpha}^* \left(\mc{N}^{\ox m}\left({\psi_{AA'}^{\ox m}}\right)
      \big\|{\psi_{A}}^{\ox m}\ox \sigma_{B^m}^u \right)-\frac{3\log v_{m,|B|}}{m}
      -\frac{2\log (m+1)^{|A|}}{m} \nb\\
\geq &\frac{1}{m}\min_{\sigma_{B^m}} D_{\alpha}^*
      \left((\mc{N}(\psi_{AA'}))^{\ox m} \big\|{\psi_{A}}^{\ox m}
      \ox \sigma_{B^m} \right)-\frac{3\log v_{m,|B|}}{m}-\frac{2\log (m+1)^{|A|}}{m} \nb\\
\stackrel{(b)}{=}  &I_{\alpha}^*(\mc{N}_{A' \rar B},\psi_{AA'})
      -\frac{3\log v_{m,|B|}}{m}-\frac{2\log (m+1)^{|A|}}{m}, \label{eq:second}
\end{align}
where $I_{\alpha}^*(\mc{N}_{A' \rar B},\psi_{AA'}):=I_{\alpha}^*(A:B)_{\mc{N}(\psi_{AA'})}$. In Eq.~\eqref{eq:second}, $(a)$ is a result of Lemma~\ref{lem:pinching-sre}, which generalizes Proposition~\ref{prop:renyid}~(\romannumeral6), and $(b)$ comes from the additivity property of Proposition~\ref{prop:renyid}~(\romannumeral4). Now, combining Eq.~\eqref{eq:zero},
Eq.~\eqref{eq:first} and Eq.~\eqref{eq:second} together, and letting $m\rar\infty$, we arrive at
\begin{align}
sc(\mc{N}_{A' \rar B},R)
&\leq \sup_{\alpha>1} \frac{\alpha-1}{\alpha}\big\{R-I_{\alpha}^*(\mc{N}_{A' \rar B},\psi_{AA'})\big\}
      +\frac{3\log v_{m,|B|}}{m}+\frac{2\log (m+1)^{|A|}}{m} \nb\\
&\rar \sup_{\alpha>1} \frac{\alpha-1}{\alpha}\big\{R-I_{\alpha}^*(\mc{N}_{A' \rar B},\psi_{AA'})\big\}.
      \label{eq:third}
\end{align}

At last, noticing that Eq.~\eqref{eq:third} holds for arbitrary pure state $\psi_{AA'}$,
we can optimize it over all pure states
\beq
\psi(\rho)_{AA'}
:=(\1_A\ox\sqrt{\rho_{A'}})(|A|\Psi_{AA'})(\1_A\ox\sqrt{\rho_{A'}}),
\quad\text{with } \rho\in\mc{S}(A').
\eeq
Therefore,
\begin{align}
sc(\mc{N}_{A' \rar B},R)
&\leq \inf_{\rho\in\mc{S}(A')}\sup_{\alpha>1} \frac{\alpha-1}{\alpha}
 \Big\{R-I_{\alpha}^*\big(\mc{N}_{A' \rar B},\psi(\rho)_{AA'}\big)\Big\} \nb\\
&  = \inf_{\rho\in\mc{S}(A')}\sup_{0<\lambda<1} \lambda
     \Big\{R-I_{\frac{1}{1-\lambda}}^*\big(\mc{N}_{A' \rar B},\psi(\rho)_{AA'}\big)\Big\}. \label{eq:fourth}
\end{align}
The function
\beq
f(\lambda, \rho)=
\lambda\Big\{R-I_{\frac{1}{1-\lambda}}^*\big(\mc{N}_{A' \rar B},\psi(\rho)_{AA'}\big)\Big\}
\eeq
is concave and continuous in $\lambda$ on the interval $(0,1)$ by Lemma~\ref{lem:convexity},
and as shown in~\cite{LiYao2021reliable} it is convex and continuous in $\rho$ on the compact convex set $\mc{S}(A')$ for any $\lambda\in(0,1)$. Thus Sion's minimax theorem applies again. This lets us obtain
\begin{align}
sc(\mc{N}_{A' \rar B},R)
&\leq\sup_{0<\lambda<1} \inf_{\rho\in\mc{S}(A')} \lambda
     \Big\{R-I_{\frac{1}{1-\lambda}}^*\big(\mc{N}_{A' \rar B},\psi(\rho)_{AA'}\big)\Big\} \nb\\
& =  \sup_{\alpha>1}\frac{\alpha-1}{\alpha}\Big\{R-I_{\alpha}^*(\mc{N}_{A'\rar B})\Big\}, \label{eq:fifth}
\end{align}
and we are done.
\end{proof}

\bigskip
The following Lemma~\ref{lem:symmetry} and Lemma~\ref{lem:convexity} are used in the
proof of Theorem~\ref{thm:main-upper}.

\begin{lemma}
\label{lem:symmetry}
The first equality of Eq.~\eqref{eq:first} holds: using the same notation as there, we have
\begin{align}
 \min_{\sigma_{B^m} \in \mc{S}(B^m)} & D_{\alpha}^{\flat}
  \left(\mc{P}_{\sigma_{B^m}^u}\circ\mc{N}^{\ox m}\left(\Psi^t_{A^mA'^m}\right)
      \big\|\pi^t_{A^m}\ox \sigma_{B^m}\right) \nb\\
=\min_{\sigma_{B^m}\in\mc{S}_{\rm{sym}}(B^m)} & D_{\alpha}^{\flat}
  \left(\mc{P}_{\sigma_{B^m}^u}\circ\mc{N}^{\ox m}\left(\Psi^t_{A^mA'^m}\right)
      \big\| \pi^t_{A^m}\ox \sigma_{B^m}\right). \label{equ:symmin}
\end{align}
\end{lemma}

\begin{proof}
Since $\mc{S}_{\rm{sym}}(B^m)\subset\mc{S}(B^m)$, the ``$\leq$'' part is obvious. In the
following, we prove the opposite direction. For a permutation $\iota \in S_m$, let
$\mc{W}^\iota_{A^m}$ be the unitary operation that acts as $\mc{W}^\iota_{A^m}(X_{A^m})
=W_{A^m}^\iota X_{A^m}W_{A^m}^{\iota \dagger }$, where $W^\iota_{A^m}$ is given in
Eq.~\eqref{eq:permutation}. Also let $\mc{W}^\iota_{{A'}^m}$ and $\mc{W}^\iota_{B^m}$
be defined similarly. It is easy to see that
\beq
 \mc{P}_{\sigma_{B^m}^u}\circ\mc{N}^{\ox m}\circ\mc{W}^\iota_{{A'}^m}
=\mc{W}^\iota_{B^m}\circ\mc{P}_{\sigma_{B^m}^u}\circ\mc{N}^{\ox m}
\eeq
holds for any permutation $\iota$. This, together with the fact that $\Psi^t_{A^mA'^m}\in\mc{S}_{\rm{sym}}(A^mA'^m)$, ensures that
\beq\label{eq:symmetry-1}
\mc{P}_{\sigma_{B^m}^u}\circ\mc{N}^{\ox m}
\left(\Psi^t_{A^mA'^m}\right)\in\mc{S}_{\rm{sym}}(A^mB^m).
\eeq
Eq.~\eqref{eq:symmetry-1} and the fact that $\pi^t_{A^m}\in\mc{S}_{\rm{sym}}(A^m)$ lead to
\begin{align}
 &D_{\alpha}^{\flat}\left(\mc{P}_{\sigma_{B^m}^u}\circ\mc{N}^{\ox m}
    \left(\Psi^t_{A^mA'^m}\right)\big\|\pi^t_{A^m}\ox \sigma_{B^m}\right) \nb\\
=&D_{\alpha}^{\flat}\left(\mc{W}^\iota_{A^m}\ox\mc{W}^\iota_{B^m}
    \left(\mc{P}_{\sigma_{B^m}^u}\circ\mc{N}^{\ox m}
    \left(\Psi^t_{A^mA'^m}\right)\right)\big\|\mc{W}^\iota_{A^m}
    \left(\pi^t_{A^m}\right)\ox \mc{W}^\iota_{B^m}\left(\sigma_{B^m}\right)\right) \nb\\
=&D_{\alpha}^{\flat}\left(\mc{P}_{\sigma_{B^m}^u}\circ\mc{N}^{\ox m}
    \left(\Psi^t_{A^mA'^m}\right)\big\|\pi^t_{A^m}\ox \mc{W}^\iota_{B^m}
    \left(\sigma_{B^m}\right)\right) \label{eq:symmetry-2}
\end{align}
for any $\iota \in S_m$ and $\sigma_{B^m}\in\mc{S}(B^m)$.
Proposition~\ref{prop:renyid}~(\romannumeral5) and Eq.~(\ref{eq:symmetry-2})
let us obtain
\begin{align}
     &D_{\alpha}^{\flat}\left(\mc{P}_{\sigma_{B^m}^u}\circ\mc{N}^{\ox m}
      \left(\Psi^t_{A^mA'^m}\right)\big\|\pi^t_{A^m}\ox \sigma_{B^m}\right)  \nb\\
  =  &\sum_{\iota \in S_m}\frac{1}{|S_m|}D_{\alpha}^{\flat}
      \left(\mc{P}_{\sigma_{B^m}^u}\circ\mc{N}^{\ox m}\left(\Psi^t_{A^mA'^m}\right)
      \big\|\pi^t_{A^m}\ox \mc{W}^\iota_{B^m}\left(\sigma_{B^m}\right)\right) \nb\\
\geq &D_{\alpha}^{\flat}\Big(\mc{P}_{\sigma_{B^m}^u}\circ\mc{N}^{\ox m}
      \left(\Psi^t_{A^mA'^m}\right)\big\|\pi^t_{A^m}\ox \sum_{\iota\in S_m}
      \frac{1}{|S_m|}\mc{W}^\iota_{B^m}\left(\sigma_{B^m}\right)\Big).
\end{align}
At last, noticing that $\sum_{\iota \in S_m}\frac{1}{|S_m|}\mc{W}^\iota_{B^m}
\left(\sigma_{B^m}\right)\in\mc{S}_{\rm{sym}}(B^m)$, we complete the proof.
\end{proof}

\medskip
\begin{lemma}
  \label{lem:convexity}
For any state $\rho_{AB}\in\mc{S}(AB)$, the function
\beq
g(\lambda) = \lambda I_{\frac{1}{1-\lambda}}^*(A:B)_\rho
\eeq
is convex and continuous on $(-1,1)$.
\end{lemma}
\begin{proof}
Let $\varrho\in\mc{S}(\mc{H})$ and $\omega\in\mc{S}(\mc{H})$ be two states that are
commutative. We consider
\beq
g_{\varrho,\omega}(\lambda):=\lambda D^*_{\frac{1}{1-\lambda}}(\varrho\|\omega).
\eeq
$g_{\varrho,\omega}(\lambda)$ is obviously continuous on $(-1,1)$. Since $\varrho$ and
$\omega$ commute, we have $D^*_\alpha(\varrho\|\omega)=D^\flat_\alpha(\varrho\|\omega)$
and the variational expression of Proposition~\ref{prop:renyid}~(\romannumeral3) applies.
So
\begin{align}
g_{\varrho,\omega}(\lambda)
&=s(\lambda)\lambda\max_{\tau\in\mc{S}_{\varrho}(\mc{H})}s(\lambda)
  \big\{D(\tau\|\omega)-\frac{1}{\lambda}D(\tau\|\varrho)\big\} \nb\\
&=\max_{\tau\in\mc{S}_{\varrho}(\mc{H})}\big\{\lambda D(\tau\|\omega)-D(\tau\|\varrho)\big\},
  \label{eq:convexity-1}
\end{align}
where $s(\lambda)=1$ for $\lambda\in(0,1)$ and $s(\lambda)=-1$ for $\lambda\in(-1,0)$. When
$\lambda=0$, Eq.~\eqref{eq:convexity-1} still holds without the intermediate step. From Eq.~\eqref{eq:convexity-1}, it is easy to verify that $g_{\varrho,\omega}(\lambda)$ is convex
on $(-1,1)$.

Now, we turn to the consideration of $g(\lambda)$. It is proved
in~\cite[Proposition 8]{HayashiTomamichel2016correlation} that for $\alpha\in[\frac{1}{2},\infty)$,
\beq
I^*_\alpha(A:B)_\rho=\frac{1}{n}D^*_\alpha
\left(\mc{P}_{\rho_A^{\ox n}\ox\sigma^u_{B^n}}\big(\rho_{AB}^{\ox n}\big)\big\|
\rho_A^{\ox n}\ox\sigma^u_{B^n}\right)+O\big(\frac{\log n}{n}\big),
\eeq
where the underlying constants in the term $O\big(\frac{\log n}{n}\big)$ are independent
of $\alpha$. Thus, we have
\beq\label{eq:convexity-2}
g(\lambda)=\frac{1}{n}\lambda D^*_{\frac{1}{1-\lambda}}
\left(\mc{P}_{\rho_A^{\ox n}\ox\sigma^u_{B^n}}\big(\rho_{AB}^{\ox n}\big)\big\|
\rho_A^{\ox n}\ox\sigma^u_{B^n}\right)+O\big(\frac{\log n}{n}\big)\lambda.
\eeq
The second term of the right hand side of Eq.~\eqref{eq:convexity-2} vanishes uniformly in
$\lambda$, as $n$ increases. So, the result obtained above for $g_{\varrho,\omega}(\lambda)$
lets us complete the proof.
\end{proof}

\section{Conclusion and discussion}
\label{sec:discussion}
We have determined the strong converse exponents for both entanglement-assisted and
quantum-feedback-assisted communication over quantum channels, building on previous
works of~\cite{GuptaWilde2015multiplicativity} and~\cite{CMW2016strong}. The formulas,
being the same for these two tasks, are expressed in terms of the sandwiched R{\'e}nyi
information $I_{\alpha}^*(\mc{N}_{A \rar B})$ of the channel $\mc{N}_{A \rar B}$
with $\alpha>1$, providing a complete operational interpretation for this quantity.
We point out that in~\cite{LiYao2021reliable}, an operational interpretation for
$I_{\alpha}^*(\mc{N}_{A \rar B})$ with $\alpha\in(1,2]$ has been found by the authors,
in characterizing the reliability function of quantum channel simulation.

Our results reinforce the viewpoint that the theory of entanglement-assisted
communication is the natural quantum generalization of Shannon's theory of classical
communication. On the one hand, the obtained strong converse exponents take the form
similar to Arimoto's exponent for classical channels~\cite{Arimoto1973converse}. On
the other hand, these results imply that additional classical or quantum feedback does
not change the strong converse exponents of entanglement-assisted communication, in
full analogy to the classical situation~\cite{Augustin1978noisy, CsiszarKorner1982feedback, DueckKorner1979reliability}, too.

At last, we comment that, to achieve the strong converse exponents for entanglement-assisted
communication, the shared entanglement can be restricted to the form of maximally entangled
states. This is in agreement with the fact that maximally entangled states suffice to
achieve the entanglement-assisted capacities~\cite{BSST2002entanglement}. Inspecting our
proof, we see that the entanglement we used is an ensemble of maximally entangled states.
However, for any fixed transmission rate $R$, a particular one from the ensemble works,
although the ensemble and the particular maximally entangled state vary with $R$.

\section*{Acknowledgements}
We are grateful to the anonymous referees for their valuable suggestions, which have
helped us improve the manuscript. In particular, one of the anonymous referees has
suggested Lemma~\ref{lem:pinching-sre} to us.

{\appendix[Auxiliary Lemmas]
The following is Sion's minimax theorem~\cite{Komiya1988elementary, Sion1958general}.
\begin{lemma}
\label{lemma:minimax}
Let $X$ be a compact convex set in a topological vector space $V$ and $Y$ be a convex
subset of a vector space $W$. Let $f : X \times Y \rar \mb{R}$ be such that
\begin{enumerate}[(i)]
  \item $f(x,\cdot)$ is quasi-concave and upper semi-continuous on $Y$ for each $x \in X$, and
  \item $f(\cdot, y)$ is quasi-convex and lower semi-continuous on $X$ for each $y \in Y$.
\end{enumerate}
Then, we have
\begin{equation}
\label{eq:mini}
\inf_{x \in X} \sup_{y \in Y} f(x,y)= \sup_{y \in Y} \inf_{x \in X} f(x,y),
\end{equation}
and the infima in Eq.~(\ref{eq:mini}) can be replaced by minima.
\end{lemma}

\medskip
The following lemma is adapted from~\cite{Holevo2000reliability} and~\cite{SchumacherWestmoreland1997sending}, and it constitutes the main technical part of the
Holevo-Schumacher-Westmoreland theorem (see also, e.g.,~\cite[Chapter 20]{Wilde2013quantum}).
\begin{lemma}
\label{lem:HSW}
Let $\{p(x),\rho_x\}_{x\in\mc{X}}$ be an ensemble of quantum states on a Hilbert space
of finite dimension, where $p$ is a probability distribution on $\mc{X}$. Let
\beq
R<\chi(\{p(x),\rho_x\}_x)\equiv\sum_xp(x)D\big(\rho_x\|\sum_xp(x)\rho_x\big)
\eeq
be fixed. Denote $\mc{S}_n=\{\rho_{x_1}\ox\cdots\ox\rho_{x_n}\}_{x^n\in\mc{X}^{\times n}}$
and let $p^n(x^n)=\prod_{i=1}^np(x_i)$ be the product distribution on $\mc{X}^{\times n}$.
For each $n\in\mathbb{N}$, there exist a set of quantum states $\mc{O}_n\equiv\{\omega_n^{(m)}\}_{m=1}^M\subseteq \mc{S}_n$ and a quantum measurement $\mc{D}_n\equiv\{\Lambda_n^{(m)}\}_{m=1}^M$ such that $M=\lfloor2^{nR}\rfloor$ and
\beq
\widetilde{P_s}((\mc{O}_n,\mc{D}_n))
:=\frac{1}{M}\sum_{m=1}^M\tr\omega_n^{(m)}\Lambda_n^{(m)}
\rar 1, \quad \text{as } n\rar\infty.
\eeq
The set $\mc{O}_n$ can be constructed by randomly choosing each $\omega_n^{(m)}$ from
$\mc{S}_n$ according to the probability distribution $p^n$.
\end{lemma}

\medskip
The following lemma is a trivial generalization of~\cite[Lemma 3]{HayashiTomamichel2016correlation},
where the case $\mc{P}=\mc{P}_\sigma$ was proven.
\begin{lemma}
\label{lem:pinching-sre}
let $\rho\in\mc{S}(\mc{H})$ be a quantum state and $\sigma\in\mc{L}(\mc{H})_+$
be positive semidefinite. Let $\{P_i\}_{i=1}^M$ be a set of projections such that $\sum_iP_i=\1$
and for each $i$, $\supp(P_i)$ is contained in an eigenspace of $\sigma$.
Then for $\alpha\geq 0$ and the CPTP map $\mc{P}:X\mapsto\sum_iP_iXP_i$, we have
\beq
D_\alpha^*(\rho\|\sigma)\leq D_\alpha^*(\mc{P}(\rho)\|\sigma)+f_\alpha(M),
\eeq
where
\beq
f_\alpha(M)=
\begin{cases}
\log M,  &\text{  if } \alpha\in[0,2], \\
2\log M, &\text{  if } \alpha>2.
\end{cases}
\eeq
\end{lemma}

The proof of~\cite[Lemma 3]{HayashiTomamichel2016correlation} works here, Simply using the pinching
inequality $\rho\leq M\mc{P}(\rho)$ instead of $\rho\leq v(\sigma)\mc{P}_\sigma(\rho)$.
}

\end{document}